\makeatletter
\let\@twosidetrue\@twosidefalse
\let\@mparswitchtrue\@mparswitchfalse
\makeatother

\RequirePackage{lineno}
\documentclass{llncs}
\usepackage[left=3.5cm,right=3.5cm,top=3.5cm,bottom=3.5cm]{geometry}
\usepackage[latin2]{inputenc}
\usepackage[english]{babel}
\usepackage{dsfont, complexity}
\usepackage{amsmath, lmodern}
\usepackage[title]{appendix}
\usepackage{amssymb}
\usepackage{lmodern}
\usepackage{float}
\usepackage{graphicx}
\usepackage{tikz}
\usepackage{tkz-graph}
\usepackage{xcolor}
\usepackage{url}
\usepackage{paralist, enumerate}
\usepackage{refcount}
\usepackage{longtable}

\usepackage{amsfonts}
\usepackage{multirow}
\usepackage{cancel}
\usepackage{soul} 
\usepackage[all]{xy}
\usepackage{subfigure}
\usepackage{color}
\usepackage{listings}
\usepackage{notoccite}
\usepackage{mathtools} 
\usepackage[normalem]{ulem}  
\usepackage{caption}
\usepackage{multicol}

\usepackage{thmtools, mathtools}
\usepackage{thm-restate}

\usepackage{algorithm}

\usepackage[textsize=footnotesize]{todonotes}

\usepackage[noend]{algpseudocode}
\usepackage{etoolbox}

\makeatletter
\newcommand*{\algrule}[1][\algorithmicindent]{%
  \makebox[#1][l]{%
    \hspace*{.2em}
    \vrule height .75\baselineskip depth .25\baselineskip
  }
}

\newcount\ALG@printindent@tempcnta
\def\ALG@printindent{%
    \ifnum \theALG@nested>0
    \ifx\ALG@text\ALG@x@notext
    \else
    \unskip
    \ALG@printindent@tempcnta=1
    \loop
    \algrule[\csname ALG@ind@\the\ALG@printindent@tempcnta\endcsname]%
    \advance \ALG@printindent@tempcnta 1
    \ifnum \ALG@printindent@tempcnta<\numexpr\theALG@nested+1\relax
    \repeat
    \fi
    \fi
}
\patchcmd{\ALG@doentity}{\noindent\hskip\ALG@tlm}{\ALG@printindent}{}{\errmessage{failed to patch}}
\patchcmd{\ALG@doentity}{\item[]\nointerlineskip}{}{}{} 
\makeatother

\algnewcommand\algorithmicforeach{\textbf{for each}}
\algdef{S}[FOR]{ForEach}[1]{\algorithmicforeach\ #1\ \algorithmicdo}

\title{An Algorithm for Strong Stability in the Student-Project Allocation Problem with Ties}

\author{
Sofiat Olaosebikan\thanks{Supported by a College of Science and Engineering Scholarship, University of Glasgow. Orcid ID: 0000-0002-8003-7887.} \and David Manlove\thanks{Supported by grant EP/P028306/1 from the Engineering and Physical Sciences Research Council. Orcid ID: 0000-0001-6754-7308.}
}
\institute{School of Computing Science, University of Glasgow, \\e-mail: \texttt{s.olaosebikan.1@research.gla.ac.uk, David.Manlove@glasgow.ac.uk}
}

\begin{document}

\maketitle              

\begin{abstract}
We study a variant of the \emph{Student-Project Allocation problem with lecturer preferences over Students} where ties are allowed in the preference lists of students and lecturers ({\sc spa-st}). We investigate the concept of \emph{strong stability} in this context. Informally, a matching is \emph{strongly stable} if there is no student and lecturer $l$ such that if they decide to form a private arrangement outside of the matching via one of $l$'s proposed projects, then neither party would be worse off and at least one of them would strictly improve. We describe the first polynomial-time algorithm to find a strongly stable matching or to report that no such matching exists, given an instance of {\sc spa-st}. Our algorithm runs in $O(m^2)$ time, where $m$ is the total length of the students' preference lists.

\end{abstract}

\thispagestyle{empty}
\setcounter{page}{1}
\pagestyle{headings}
\section{Introduction}
\label{introduction}
\label{sect:introduction}
Matching problems, which generally involve the assignment of a set of agents to another set of agents based on preferences, have wide applications in many real-world settings, including, for example, allocating junior doctors to hospitals \cite{Rot84} and assigning students to projects \cite{KIMS15}. In the context of assigning students to projects, each project is proposed by one lecturer and each student is required to provide a preference list over the available projects that she finds acceptable. Also, lecturers may provide a preference list over the students that find their projects acceptable, and/or over the projects that they propose. Typically, each project and lecturer have a specific capacity denoting the maximum number of students that they can accommodate. The goal is to find a \emph{matching}, i.e., an assignment of students to projects that respects the stated preferences, such that each student is assigned at most one project, and the capacity constraints on projects and lecturers are not violated --- the so-called \emph{Student-Project Allocation problem} ({\sc spa}) \cite{AIM07,CFG17,Man13}. 

Two major models of {\sc spa} exist in the literature: one permits preferences only from the students \cite{KIMS15}, while the other permits preferences from the students and lecturers \cite{AIM07,Kaz02}. In the latter case, three different variants have been studied based on the nature of the lecturers' preference lists. These include {\scriptsize SPA} with lecturer preferences over (i) students \cite{Kaz02}, (ii) projects \cite{IMY12,MMO18,MO08}, and (iii) (student, project) pairs \cite{AM09}. Outwith assigning students to projects, applications of each of these three variants can be seen in multi-cell networks where the goal is to find a stable assignment of users to channels at base-stations \cite{BBA19,BBAHD19,BBESA19}. 

In this work, we will concern ourselves with variant (i), i.e., the \emph{Student-Project Allocation problem with lecturer preferences over Students} ({\sc spa-s}). In this context, it has been argued in \cite{Rot84} that a natural property for a matching to satisfy is that of \emph{stability}. Informally, a \emph{stable matching} ensures that no student and lecturer would have an incentive to deviate from their current assignment. Abraham \emph{et al.}~\cite{AIM07} described two linear-time algorithms to find a stable matching in an instance of {\sc spa-s} where the preference lists are strictly ordered. In their paper, they also proposed an extension of {\sc spa-s} where the preference lists may include ties, which we refer to as the \emph{Student-Project Allocation problem with lecturer preferences over Students with Ties} ({\sc spa-st}). 

If we allow ties in the preference lists of students and lecturers, three different stability definitions are possible \cite{Irv94,IMS00,IMS03}. We give an informal definition in what follows. Suppose that $M$ is a matching in an instance of {\sc spa-st}.  Then $M$ is (i) weakly stable, (ii) strongly stable, or (iii) super-stable, if there is no student and lecturer $l$ such that if they decide to become assigned outside of $M$ via one of $l$'s proposed projects, respectively,
\begin{enumerate}[(i)]
\item both of them would strictly improve, 
\item one of them would be better off, and the other would not be worse off
\item neither of them would be worse off.
\end{enumerate}
These concepts were first defined and studied by Irving \cite{Irv94} in the context of the \emph{Stable Marriage problem with Ties} ({\sc smt}) (the restriction of {\sc spa-st} in which each lecturer offers only one project, the capacity of each project and lecturer is $1$, the numbers of students and lecturers are equal, each student finds all projects acceptable, and each lecturer finds all students acceptable). It was subsequently extended to the \emph{Hospitals/Residents problem with Ties} ({\scriptsize HRT}) \cite{IMS00,IMS03} (where {\scriptsize HRT} is the special case of {\sc spa-st} in which each lecturer offers only one project, and the capacity of each project is the same as the capacity of the lecturer offering the project).

\vspace{2mm}
\noindent
{\bf Existing results in {\sc spa-st. }}Every instance of {\sc spa-st} admits a weakly stable matching, which could be of different sizes \cite{MIIMM02}. Moreover, the problem of finding a maximum size weakly stable matching ({\scriptsize MAX-SPA-ST}) is NP-hard \cite{IMMM99,MIIMM02}, even for the so-called \emph{Stable Marriage problem with Ties and Incomplete lists} ({\sc smti}). Cooper and Manlove \cite{CM18} described a $\frac{3}{2}$-approximation algorithm for {\scriptsize MAX-SPA-ST}. On the other hand, Irving \emph{et al}.~argued in \cite{IMS00} that super-stability is a natural and most robust solution concept to seek in cases where agents have incomplete information. Recently, Olaosebikan and Manlove \cite{OM18} showed that if an instance of {\sc spa-st} admits a super-stable matching $M$, then all weakly stable matchings in the instance are of the same size (equal to the size of  $M$), and match exactly the same set of students. The main result of their paper was a polynomial-time algorithm to find a super-stable matching or report that no such matching exists, given an instance of {\sc spa-st}. Their algorithm runs in $O(L)$ time, where $L$ is the total length of all the preference lists.

\vspace{2mm}
\noindent
\textbf{Motivation.} It was motivated in \cite{IMS03} that weakly stable matching may be undermined by bribery or persuasion, in practical applications of {\sc hrt}. In what follows, we give a corresponding argument for an instance $I$ of {\sc spa-st}. Suppose that $M$ is a weakly stable matching in $I$, and suppose that a student $s_i$ prefers a project $p_j$ (where $p_j$ is offered by lecturer $l_k$) to her assigned project in $M$, say $p_{j'}$ (where $p_{j'}$ is offered by a lecturer different from $l_k$). Suppose further that $p_j$ is full and $l_k$ is indifferent between $s_i$ and one of the worst student/s assigned to $p_j$ in $M$, say $s_{i'}$. Clearly, the pair $(s_i, p_j)$ does not constitute a blocking pair for the weakly stable matching $M$, as $l_k$ would not improve by taking on $s_i$ in the place of $s_{i'}$. However, $s_i$ might be overly invested in $p_j$ that she is even ready to persuade or even bribe $l_k$ to reject $s_{i'}$ and accept her instead; $l_k$ being indifferent between $s_i$ and $s_{i'}$ may decide to accept $s_i$'s proposal. We can reach a similar argument if the roles are reversed. However, if $M$ is strongly stable, it cannot be potentially undermined by this type of (student, project) pair.

Henceforth, if a {\sc spa-st} instance admits a strongly stable matching, we say that such an instance is solvable. Unfortunately not every instance of {\sc spa-st} is solvable. To see this, consider the case where there are two students, two projects and two lecturers, the capacity of each project and lecturer is $1$, the students have exactly the same strictly ordered preference list of length 2, and each of the lecturers preference list is a single tie of length 2 (any matching will be undermined by a student and lecturer that are not matched together). However, it should be clear from the discussions above that in cases where a strongly stable matching exists, it should be preferred over a matching that is merely weakly stable. 

\vspace{2mm}
\noindent
\textbf{Related work.} The following are previous results for strong stability in the literature. Irving \cite{Irv94} gave an $O(n^4)$ algorithm for computing strongly stable matchings in an instance of {\sc smt}, where $n$ is the number of men (equal to the number of women). This algorithm was subsequently extended by Manlove \cite{Man99} to instances of {\sc smti}, which is a generalisation of {\sc smt} for which the preference lists need not be complete. The extended algorithm also has running time $O(n^4)$. Irving \emph{et al.~}\cite{IMS03} described an algorithm to find a strongly stable matching or report that no such matching exists, given an instance of {\scriptsize HRT}. The algorithm has running time $O(m^2)$, where $m$ is the total number of acceptable (resident, hospital) pairs. Subsequently, Kavitha \emph{et al.~}\cite{KMMP04} presented two strong stability algorithms with improved running time; one for {\sc smti} with running time $O(nm)$, where $n$ is the number of vertices and $m$ is the number of edges, and the other for {\sc hrt} with running time $O(m \sum_{h \in H} p_h)$, where $H$ is the set of all hospitals and $p_h$ is the capacity of a hospital $h$. These two algorithms build on the ones described in \cite{Irv94,IMS03,Man99}. A recent result in strong stability is the work of Kunysz \cite{Kun18}, where he described an $O(nm \log (Wn))$ algorithm for computing a maximum weight strongly stable matching given an instance of {\sc smti}, where $W$ is the maximum weight of an edge.

\vspace{2mm}
\noindent
\textbf{Our contribution.} We present the first polynomial-time algorithm to find a strongly stable matching or report that no such matching exists, given an instance of {\sc spa-st} -- thus solving an open problem given in \cite{AIM07,OM18}.  Our algorithm is student-oriented, which implies that if the given instance is solvable then our algorithm will output a solution in which each student has at least as good a project as she could obtain in any strongly stable matching. We note that our algorithm is a non-trivial extension of the strong stability algorithms for {\sc smt} \cite{Irv94}, {\sc smti} \cite{Man99}, and {\sc hrt} \cite{IMS03} (we discuss this further in Sect.~\ref{spa-st-strong-difference}).

The remainder of this paper is structured as follows. We give a formal definition of the {\sc spa-s} problem, the {\sc spa-st} variant, and the three stability concepts in Sect.~\ref{sect:definitions}. We give some intuition for the strong stability definition in Sect.~\ref{justification}. We describe our algorithm for {\sc spa-st} under strong stability in Sect.~\ref{sect:spa-st-strong-algorithm}. Further, in Sect.~\ref{sect:spa-st-strong-algorithm}, we also illustrate an execution of our algorithm with respect to an instance of {\sc spa-st} before moving on to present the algorithm's correctness and complexity results, along with proof of correctness. Finally, we present some potential direction for future work in Sect.~\ref{sect:spa-st-strong-conclusions}.

\vspace{-2mm}
\section{Preliminary definitions}
\label{sect:definitions}
In this section, we give a formal definition of {\sc spa-s} as described in the literature \cite{AIM07}. We also give a formal definition of {\sc spa-st} --- a generalisation of {\sc spa-s} in which preference lists can include ties. 
\vspace{-2mm}
\subsection{Formal definition of {\sc spa-s}}
\label{subsect:spa-s}
An instance $I$ of {\sc spa-s} involves a set $\mathcal{S} = \{s_1 , s_2, \ldots , s_{n_1}\}$ of \emph{students}, a set $\mathcal{P} = \{p_1 , p_2, \ldots , p_{n_2}\}$ of \emph{projects} and a set $\mathcal{L} = \{l_1 , l_2, \ldots , l_{n_3}\}$ of \emph{lecturers}. Each student $s_i$ ranks a subset of $\mathcal{P}$ in strict order, which forms her preference list. We say that $s_i$ finds $p_j$ \emph{acceptable} if  $p_j$ appears on $s_i$'s preference list. We denote by $A_i$ the set of projects that $s_i$ finds acceptable. 

Each lecturer $l_k \in \mathcal{L}$ offers a non-empty set of projects $P_k$, where $P_1, P_2, \ldots,$ $P_{n_3}$ partitions $\mathcal{P}$, and $l_k$ provides a preference list, denoted by $\mathcal{L}_k$, ranking in strict order of preference those students who find at least one project in $P_k$ acceptable. Also $l_k$ has a capacity $d_k \in \mathbb{Z}^+$, indicating the maximum number of students she is willing to supervise. Similarly each project $p_j \in \mathcal{P}$ has a capacity $c_j \in \mathbb{Z}^+$ indicating the maximum number of students that it can accommodate.

We assume that for any lecturer $l_k$, $\max\{c_j: p_j \in P_k\} \leq d_k \leq \sum \{c_j: p_j \in P_k\}$ (i.e., the capacity of $l_k$ is (i) at least the highest capacity of the projects offered by $l_k$, and (ii) at most the sum of the capacities of all the projects $l_k$ is offering). We denote by $\mathcal{L}_k^j$, the \emph{projected preference list} of lecturer $l_k$ for $p_j$, which can be obtained from $\mathcal{L}_k$ by removing those students that do not find $p_j$ acceptable (thereby retaining the order of the remaining students from $\mathcal{L}_k$).

An \emph{assignment} $M$ is a subset of $\mathcal{S} \times \mathcal{P}$ such that $(s_i, p_j) \in M$ implies that $s_i$ finds $p_j$ acceptable. If $(s_i, p_j) \in M$, we say that $s_i$ \emph{is assigned to} $p_j$, and $p_j$ \emph{is assigned} $s_i$. For convenience, if $s_i$ is assigned in $M$ to $p_j$, where $p_j$ is offered by $l_k$, we may also say that $s_i$ \emph{is assigned to} $l_k$, and $l_k$ \emph{is assigned} $s_i$. For any project $p_j \in \mathcal{P}$, we denote by $M(p_j)$ the set of students assigned to $p_j$ in $M$. Project $p_j$ is \emph{undersubscribed}, \emph{full} or \emph{oversubscribed} according as $|M(p_j)|$ is less than, equal to, or greater than $c_j$, respectively. Similarly, for any lecturer $l_k \in \mathcal{L}$, we denote by $M(l_k)$ the set of students assigned to $l_k$ in $M$. Lecturer $l_k$ is \emph{undersubscribed}, \emph{full} or \emph{oversubscribed} according as $|M(l_k)|$ is less than, equal to, or greater than $d_k$, respectively. A \emph{matching} $M$ is an assignment such that  $|M(s_i)|\leq 1$, $|M(p_j)|\leq c_j$ and $|M(l_k)|\leq d_k$. If $s_i$ is assigned to some project in $M$, for convenience we let $M(s_i)$ denote that project. 

\subsection{Ties in the preference lists}
\label{sect:spa-st-strong}
We now give a formal definition, similar to the one given in \cite{OM18}, for the generalisation of {\sc spa-s} in which the preference lists can include ties. In the preference list of lecturer $l_k\in \mathcal L$, a set $T$ of $r$ students forms a \emph{tie of length $r$} if $l_k$ does not prefer $s_i$ to $s_{i'}$ for any $s_i, s_{i'} \in T$ (i.e., $l_k$ is \emph{indifferent} between $s_i$ and $s_{i'}$). A tie in a student's preference list is defined similarly. For convenience, in what follows we consider a non-tied entry in a preference list as a tie of length one. We denote by {\sc spa-st} the generalisation of {\sc spa-s} in which the preference list of each student (respectively lecturer) comprises a strict ranking of ties, each comprising one or more projects (respectively students).
An example {\sc spa-st} instance $I_1$ is given in Fig.~\ref{fig:spa-st-strong-1}, which involves the set of students $\mathcal{S} = \{s_1, s_2, s_3\}$, the set of projects $\mathcal{P} = \{p_1, p_2, p_3\}$ and the set of lecturers $\mathcal{L} = \{l_1, l_2\}$.  Ties in the preference lists are indicated by round brackets.

\begin{figure}[h]
\centering
\footnotesize
\begin{tabular}{llll}
\hline
\texttt{Student preferences} & \qquad \qquad  & \texttt{Lecturer preferences} \\ 
$s_1$: \;($p_1$ \; $p_2$)  &  & $l_1$: \; $s_3$ \;($s_1$ \; $s_2$) & $l_1$ offers $p_1$, $p_2$\\ 
$s_2$: \; $p_2$ \; $p_3$  &  & $l_2$: \; ($s_3$ \; $s_2$) & $l_2$ offers $p_3$\\ 
$s_3$: \; $p_3$ \; $p_1$ &  & &\\
 &  & Project capacities: $c_1 = c_2 = c_3 = 1$& \\
 &  & Lecturer capacities: $d_1 = 2, \; d_2 = 1$&\\ 
\hline
\end{tabular}
\caption{\label{fig:spa-st-strong-1} \small An example {\sc spa-st} instance $I_1$. }
\end{figure}
\noindent

In the context of {\sc spa-st}, we assume that all notation and terminology carries over from {\sc spa-s} with the exception of stability, which we now define. When ties appear in the preference lists, three types of stability arise, namely \emph{weak stability, strong stability and super-stability} \cite{IMS00,IMS03}. In what follows, we give a formal definition of these three stability concepts in the context of {\sc spa-st}. Henceforth, $I$ is an instance of {\sc spa-st}, $(s_i, p_j)$ is an acceptable pair in $I$ and $l_k$ is the lecturer who offers $p_j$. 
\begin{definition}[weak stability \cite{CM18}]
\label{def:weak-stability}
\normalfont
Let $M$ be a matching in $I$. We say that $M$ is \emph{weakly stable} if it admits no blocking pair, where a \emph{blocking pair} of $M$ is an acceptable pair $(s_i, p_j) \in (\mathcal{S} \times \mathcal{P}) \setminus M$ such that (a) and (b) holds as follows:
\begin{enumerate}[(a)]
    \item either $s_i$ is unassigned in $M$ or $s_i$ prefers $p_j$ to $M(s_i)$;
    \item either (i), (ii), or (iii) holds as follows:
        \begin{enumerate}[(i)]
        \item $p_j$ is undersubscribed and $l_k$ is undersubscribed;
        \item  $p_j$ is undersubscribed, $l_k$ is full and either $s_i \in M(l_k)$, or $l_k$ prefers $s_i$ to the worst student/s in $M(l_k)$;
        \item  $p_j$ is full and $l_k$ prefers $s_i$ to the worst student/s in $M(p_j)$.
        \end{enumerate}
\end{enumerate}

\end{definition}

\begin{definition}[super-stability \cite{OM18}]
\label{def:super-stability}
\normalfont
We say that $M$ is \emph{super-stable} if it admits no blocking pair, where a \emph{blocking pair} of $M$ is an acceptable pair $(s_i, p_j) \in (\mathcal{S} \times \mathcal{P}) \setminus M$ such that (a) and (b) holds as follows:
\begin{enumerate}[(a)]
    \item either $s_i$ is unassigned in $M$, or $s_i$ prefers $p_j$ to $M(s_i)$ or is indifferent between them;
    \item either (i), (ii), or (iii) holds as follows:
    \begin{enumerate} [(i)]
    \item $p_j$ is undersubscribed and $l_k$ is undersubscribed;
    \item $p_j$ is undersubscribed, $l_k$ is full, and either $s_i \in M(l_k)$ or $l_k$  prefers $s_i$ to the worst student/s in $M(l_k)$ or is indifferent
    between them;
    \item $p_j$ is full and $l_k$  prefers $s_i$ to the worst student/s in $M(p_j)$ or is indifferent between them.
    \end{enumerate}
\end{enumerate}
\end{definition}
\begin{definition}[strong stability]
\label{def:strong-stability}
\normalfont
We say that $M$ is \emph{strongly stable} in $I$ if it admits no blocking pair, where a \emph{blocking pair} of $M$ is an acceptable pair $(s_i, p_j) \in (\mathcal{S} \times \mathcal{P}) \setminus M$ such that either (1a and 1b) or (2a and 2b) holds as follows:
\begin{enumerate}
\item [(1a)] either $s_i$ is unassigned in $M$, or $s_i$ prefers $p_j$ to $M(s_i)$;
\item [(1b)] either (i), (ii), or (iii) holds as follows:
    \begin{enumerate}[(i)]
    \item $p_j$ is undersubscribed and $l_k$ is undersubscribed;
    \item $p_j$ is undersubscribed, $l_k$ is full, and either $s_i \in M(l_k)$ or $l_k$  prefers $s_i$ to the worst student/s in $M(l_k)$ or is indifferent between them;
    \item $p_j$ is full and $l_k$  prefers $s_i$ to the worst student/s in $M(p_j)$ or is indifferent between them.
    \end{enumerate}
\end{enumerate}

\begin{enumerate}
\item [(2a)] $s_i$ is indifferent between $p_j$ and $M(s_i)$;
\item [(2b)] either (i), (ii), or (iii) holds as follows:
    \begin{enumerate}[(i)]
    \item $p_j$ is undersubscribed, $l_k$ is undersubscribed and $s_i \notin M(l_k)$;
    \item $p_j$ is undersubscribed, $l_k$ is full, $s_i \notin M(l_k)$, and $l_k$ prefers $s_i$ to the worst student/s in $M(l_k)$;
    \item $p_j$ is full and $l_k$ prefers $s_i$ to the worst student/s in $M(p_j)$.
    \end{enumerate}
\end{enumerate}
\end{definition}

\vspace{2mm}
In the remainder of this paper, any usage of the term \emph{blocking pair} refers to the version of this term for strong stability as defined in Definition \ref{def:strong-stability}. We give an intuition behind the strong stability definition is what follows.

\section{Justification of the strong stability definition}
\label{justification}
It should be clear from our definition of a blocking pair $(s_i, p_j)$ that if $s_i$ seeks to become assigned to $p_j$ outside of $M$, then at most one of $s_i$ and $l_k$ can be indifferent to the switch, whilst at least one of them must strictly improve. In what follows, we justify our definition in more detail (we remark that some of the argument is similar to that given for the blocking pair definition in the {\sc spa-s} case \cite[Sect 2.2]{AIM07}).

We consider the first part of the definition. In Definition \ref{def:strong-stability}(1a); clearly if the assignment between $s_i$ and $p_j$ is permitted, $s_i$ will improve relative to $M$. Now, let us consider $l_k$'s perspective. In Definition \ref{def:strong-stability}(1b)(i), $l_k$ will be willing to take on $s_i$ for $p_j$, since there is a free space. In Definition \ref{def:strong-stability}(1b)(ii), if $s_i$ was already assigned in $M$ to a project offered by $l_k$ then $l_k$ will agree to the switch, since the total number of students assigned to $l_k$ remains the same. However, if $s_i$ was not already assigned in $M$ to a project offered by $l_k$, since $l_k$ is full, $l_k$ will need to reject some student assigned to her in order to take on $s_i$. Obviously, $l_k$ will not reject a student that she prefers to $s_i$; thus $l_k$ will either improve or be no worse off after the switch. Finally, in Definition \ref{def:strong-stability}(1b)(iii), since $p_j$ is full, $l_k$ will need to reject some student assigned to $p_j$ in order to take on $s_i$.  Again, $l_k$ will either improve or be no worse off after the switch. Under this definition, as observed in \cite[Sect 2.2]{AIM07}, if $s_i$ was already assigned  in $M$ to a project offered by $l_k$, then the number of students assigned to $l_k$ will decrease by $1$ (the reason for this was further justified in  \cite[Sect 6.1]{AIM07}).

Next, we consider the second part of the definition. In Definition \ref{def:strong-stability}(2a), if the assignment between $s_i$ and $p_j$ is permitted, clearly $s_i$ will be no worse off after the switch. Again, we consider $l_k$'s perspective (we note that in this case, $l_k$ must improve after the switch.). In Definition \ref{def:strong-stability}(2b)(i), if $p_j$ and $l_k$ is undersubscribed, then the only way that $l_k$ would improve is if $s_i$ is not already assigned in $M$ to a project offered by $l_k$. If this is the case, then $l_k$ will agree to the switch since there is a free space and she will get one more student to supervise, namely $s_i$.  In Definition \ref{def:strong-stability}(2b)(ii), if $p_j$ is undersubscribed and $l_k$ is full, the only way $l_k$ could improve is first for $s_i$ to not be assigned in $M$ to a project offered by $l_k$. If this is the case then $l_k$ will need to reject some student assigned to her in $M$ in order to take on $s_i$. Obviously, $l_k$ will be willing to reject a student that is worse than $s_i$ on her list. Similarly, in Definition \ref{def:strong-stability}(2b)(iii), if $p_j$ is full in $M$, $l_k$ will need to reject some student assigned to $p_j$ in $M$ in order to take on $s_i$. Clearly, $l_k$ must prefer $s_i$ to such student so that $l_k$ can have a better set of students assigned to $p_j$ after the switch. We remark that if $s_i$ is already assigned in $M$ to a project offered by $l_k$, the number of students assigned to $l_k$ in $M$ will decrease by $1$ after the switch. 

We illustrate this using the example {\sc spa-st} instance $I_2$ in Fig.~\ref{fig:spa-st-strong-2}. Clearly, $M_1 = \{(s_1, p_2), $ $(s_2, p_1)\}$ is a matching in $I_2$. However, by satisfying the blocking pair $(s_1, p_1)$, we obtain the matching $M_2 = \{(s_1,p_1)\}$, which is strongly stable. In going from $M_1$ to $M_2$, lecturer $l_1$ got a better student to take on $p_1$, with $l_1$ having to lose one student, namely $s_2$. Following a similar argument as in \cite[Sect 6.1]{AIM07} for the {\sc spa-s} case, in practice, $l_1$ might not agree to take on $s_1$ for $p_1$, since in doing so $l_1$ loses a student. So, one could relax Definition \ref{def:strong-stability}(2b)(iii) to prevent such change from happening. Further, we argue that this could lead to two new problems.

\begin{enumerate}
    \item We introduce an element of strategy into the problem by allowing a matching such as $M_1$ to be strongly stable. That is, a student could provide a shorter preference list in order to obtain a more suitable project. To illustrate this, suppose that $s_1$ has only listed $p_1$ in the example instance $I_2$. Irrespective of how we relax Definition \ref{def:strong-stability}(2b)(iii), $s_i$ would be assigned to $p_1$. However, this strategy might not be beneficial for a student in all cases, since by not listing all of her acceptable projects, a student is at risk of not being assigned in the final strongly stable matching.
    \item If we allow both $M_1$ and $M_2$ to be strongly stable, this would imply that the instance admits strongly stable matchings of different sizes. Hence, we would seek a maximum size strongly stable matching in order to match as many students to projects as possible. However, we conjecture that this problem is NP-hard, following from related problems of finding maximum size stable matching in the literature.
\end{enumerate}

\begin{figure}[t]
\centering
\small
\begin{tabular}{llll}
\hline
\texttt{Student preferences} & \qquad \qquad  & \texttt{Lecturer preferences} \\ 
$s_1$: \;($p_1$ \; $p_2$)  &  & $l_1$: \; $s_1$ \; $s_2$ & $l_1$ offers $p_1$, $p_2$\\ 
$s_2$: \; $p_1$  &  &  &\\ 
 &  & Project capacities: $c_1 = c_2 = 1$& \\
 &  & Lecturer capacities: $d_1 = 2$&\\ 
\hline
\end{tabular}
\caption{\label{fig:spa-st-strong-2} \small An example {\sc spa-st} instance $I_2$.}
\end{figure}

\section{An algorithm for {\sc spa-st} under strong stability}
\label{sect:spa-st-strong-algorithm}
In this section we present our algorithm for {\sc spa-st} under strong stability, which we will refer to as \texttt{Algorithm SPA-ST-strong}. In Sect.~\ref{spa-st-strong-algorithm-definition}, we give some definitions relating to the algorithm. In Sect.~\ref{spa-st-strong-algorithm-description}, we give a description of our algorithm and present it in pseudocode form. In Sect.~\ref{spa-st-strong-difference}, we briefly describe the non-trivial modifications that are involved in extending the existing strong stability algorithms for {\sc smt} \cite{Irv94}, {\sc smti} \cite{Man99} and {\sc hrt} \cite{IMS03} to our algorithm for the {\sc spa-st} case. We illustrate an execution of our algorithm with respect to a {\sc spa-st} instance in Sect.~\ref{spa-st-strong-example-execution}. Finally, in Sect.~\ref{spa-st-strong-correctness-result}, we present our algorithm's correctness results along with proof of correctness.

\subsection{Definitions relating to the algorithm}
\label{spa-st-strong-algorithm-definition}
Given a pair $(s_i, p_j) \in M$, for some strongly stable matching $M$ in $I$, we call $(s_i, p_j)$ a \emph{strongly stable pair}. During the execution of the algorithm, students become \emph{provisionally assigned} to projects (and implicitly to lecturers), and it is possible for a project (and lecturer) to be provisionally assigned a number of students that exceeds its capacity. We describe a project (respectively lecturer) as \emph{replete} if at any time during the execution of the algorithm it has been full or oversubscribed. We say that a project (respectively lecturer) is \emph{non-replete} if it is not replete.

As stated earlier, for a project $p_j$, it is possible that $d_G(p_j) > c_j$ at some point during the algorithm's execution. Thus, we denote by $q_j = \min\{c_j, d_G(p_j)\}$ the \emph{quota of $p_j$} in $G$, which is the minimum between $p_j$'s capacity and the number of students provisionally assigned to $p_j$ in $G$. Similarly, for a lecturer $l_k$, it is possible that $d_G(l_k) > d_k$ at some point during the algorithm's execution. At this point, we denote by $\alpha_k = \sum \{q_j: p_j \in P_k \cap P\}$ the total quota of projects offered by $l_k$ that is provisionally assigned to students in $G$ and we denote by $q_k = \min\{d_k, d_G(l_k), \alpha_k\}$ the \emph{quota of $l_k$} in $G$.

The algorithm proceeds by deleting from the preference lists certain $(s_i, p_j)$ pairs that are not strongly stable. By the term \emph{delete} $(s_i, p_j)$, we mean the removal of $p_j$ from $s_i$'s preference list and the removal of $s_i$ from $\mathcal{L}_k^j$ (the projected preference list of lecturer $l_k$ for $p_j$); in addition, if $(s_i, p_j) \in E$ we delete the edge from $G$. By the \emph{head} and \emph{tail} of a preference list at a given point we mean the first and last tie respectively on that list after any deletions might have occurred (recalling that a tie can be of length $1$).  Given a project $p_j$, we say that a student $s_i$ is \emph{dominated in} $\mathcal{L}_k^j$ if $s_i$ is worse than at least $c_j$ students who are provisionally assigned to $p_j$. The concept of a student becoming dominated in a lecturer's preference list is defined in a slightly different manner.

\begin{definition}[Dominated in $\mathcal{L}_k$]
\normalfont
At a given point during the algorithm's execution, let $\alpha_k$ and $d_G(l_k)$ be as defined above. We say that a student $s_i$ is \emph{dominated in} $\mathcal{L}_k$ if $\min\{d_G(l_k), \alpha_k\} \geq d_k$\footnote{If a student $s_i$ is provisionally assigned in $G$ to two different projects offered by $l_k$ then, potentially, $d_G(l_k) < \alpha_k$. Thus, it is important that we take the minimum of these two parameters, to avoid deleting strongly stable pairs.}, 
and $s_i$ is worse than at least $d_k$ students who are provisionally assigned in $G$ to a project offered by $l_k$.
\end{definition}
\begin{definition}[Lower rank edge]
\normalfont
We define an edge $(s_i, p_j) \in E$ as a \emph{lower rank edge} if $s_i$ is in the tail of $\mathcal{L}_k$ and $\min\{d_G(l_k), \alpha_k\} > d_k$.
\end{definition}
\begin{definition}[Bound]
\normalfont
Given an edge $(s_i, p_j) \in E$, we say that $s_i$ is \emph{bound to} $p_j$ if (i) and (ii) holds as follows:
 \begin{enumerate}[(i)]
     \item $p_j$ is not oversubscribed or $s_i$ is not in the tail of $\mathcal{L}_k^j$ (or both);
     \item $(s_i, p_j)$ is not a lower rank edge or $s_i$ is not in the tail of $\mathcal{L}_k$ (or both).
 \end{enumerate}
If $s_i$ is bound to $p_j$, we may also say that $(s_i, p_j)$ is a \emph{bound edge}. Otherwise, we refer to it as an \emph{unbound edge}.\footnote{An edge $(s_i, p_j) \in E$ can change state from \emph{bound} to \emph{unbound}, but not vice versa.}
\end{definition}
We form a \emph{reduced assignment graph} $G_r = (S_r, P_r, E_r)$ from a provisional assignment graph $G$ as follows. For each edge $(s_i, p_j) \in E$ such that $s_i$ is bound to $p_j$, we remove the edge $(s_i, p_j)$ from $G_r$\footnote{We note that we only remove this edge to form $G_r$, we do not delete the edge from $G$.} and we reduce the quota of $p_j$ in $G_r$ (and intuitively $l_k$\footnote{If $s_i$ is bound to more than one projects offered by $l_k$, for all the bound edges involving $s_i$ and these projects that we remove from $G_r$, we only reduce $l_k$'s quota in $G_r$ by one.}) by one. Further, we remove all other unbound edges incident to $s_i$ in $G_r$. Each isolated student vertex is then removed from $G_r$. Finally, if the quota of any project is reduced to $0$, or $p_j$ becomes an isolated vertex, then $p_j$ is removed from $G_r$. For each surviving $p_j$ in $G_r$, we denote by $q_j^*$ the \emph{revised quota of $p_j$}, where $q_j^*$ is the difference between $p_j$'s quota in $G$ (i.e., $q_j$) and the number of students that are bound to $p_j$. Similarly, we denote by $q_k^*$ the \emph{revised quota of $l_k$} in $G_r$, where $q_k^*$ is the difference between $l_k$'s quota in $G$ (i.e., $q_k$) and the number of students that are bound to a project offered by $l_k$. Further, for each $l_k$ who offers at least one project in $G_r$, we let $n = \sum\{q_j^* : p_j \in P_k \cap P_r\} - q_k^*$, where $n$ is the difference between the total revised quota of projects in $G_r$ that are offered by $l_k$ and the revised quota of $l_k$ in $G_r$. Now, if $n \leq 0$, we do nothing; otherwise, we extend $G_r$ as follows. We add $n$ dummy student vertices to $S_r$. For each of these dummy vertex, say $s_{d_i}$, and for each project $p_j \in P_k \cap P_r$ that is adjacent to a student vertex in $S_r$ via a lower rank edge, we add the edge $(s_{d_i}, p_j)$ to $E_r$.\footnote{An intuition as to why we add dummy students to $G_r$ is as follows. Given a lecturer $l_k$ whose project is provisionally assigned to a student in $G_r$. If $q_k^* < \sum\{q_j^* : p_j \in P_k \cap P_r\}$, then we need $n$ dummy students to offset the difference between $\sum\{q_j^* : p_j \in P_k \cap P_r\}$ and $q_k^*$, so that we don't oversubscribe $l_k$ in any maximum matching obtained from $G_r$.}

Given a set $X \subseteq S_r$ of students, define $\mathcal{N}(X)$, the \emph{neighbourhood of} $X$, to be the set of project vertices adjacent in $G_r$ to a student in $X$. If for all subsets $X$ of $S_r$, each student in $X$ can be assigned to one project in $\mathcal{N}(X)$, without exceeding the revised quota of each project in $\mathcal{N}(X)$ (i.e., $|X| \leq \sum \{q_j^*: p_j \in \mathcal{N}(X)\}$ for all $X \subseteq S_r$); then we say $G_r$ admits a \emph{perfect matching} that saturates $S_r$.

\begin{definition}[Critical set] \label{def:critical-set}
\normalfont 
 It is well known in the literature \cite{Liu68} that if $G_r$ does not admit a perfect matching that saturates $S_r$, then there must exist a \emph{deficient} subset $Z \subseteq S_r$ such that  $|Z| > \sum \{q_j^*: p_{j} \in \mathcal{N}(Z)\}$. To be precise, the \emph{deficiency} of $Z$ is defined by $\delta(Z) = |Z| - \sum \{q_j^*: p_{j} \in \mathcal{N}(Z)\}$. The \emph{deficiency} of $G_r$, denoted $\delta(G_r)$, is the maximum deficiency taken over all subsets of $S_r$. Thus, if $\delta(Z) = \delta(G_r)$, we say that $Z$ is a \emph{maximally deficient} subset of $S_r$, and we refer to $Z$ as a \emph{critical set}.
\end{definition}
We denote by $P_R$ the set of replete projects in $G$ and we denote by $P_R^*$ a subset of projects in $P_R$ which is obtained as follows. For each project $p_j \in P_R$, let $l_k$ be the lecturer who offers $p_j$. For each student $s_i$ such that $(s_i, p_j)$ has been deleted, we add $p_j$ to $P_R^*$ if (i) and (ii) holds as follows:
\begin{enumerate}[(i)]
\item either $s_i$ is unassigned in $G$, or $(s_i, p_{j'}) \in G$ where $s_i$ prefers $p_j$ to $p_{j'}$, or $(s_i, p_{j'}) \in G$ and $s_i$ is indifferent between $p_j$ and $p_{j'}$ where $p_{j'} \notin P_k$;
\item either $l_k$ is undersubscribed, or $l_k$ is full and either $s_i \in G(l_k)$ or $l_k$ prefers $s_i$ to some student assigned in $G(l_k)$.
\end{enumerate}
\begin{definition}[Feasible matching] \label{def:feasible-matching}
\normalfont
A \emph{feasible matching} in the final provisional assignment graph $G$ is a matching $M$ which is obtained as follows:
\begin{enumerate}
\item Let $G^*$ be the subgraph of $G$ induced by the students who are adjacent to a project in $P_R^*$. First, find a maximum matching $M^*$ in $G^*$;\footnote{At the point in the algorithm when we need to construct the feasible matching $M$ from $G$, if $P_R^*$ is non-empty, this phase ensures that we fill up all of the projects in $P_R^*$ to avoid a potential blocking pair involving some student that has been rejected by some project in $P_R^*$.}
\item Using $M^*$ as an initial solution, find a maximum matching $M$ in $G$.
\end{enumerate}
\end{definition}
\subsection{Description of the algorithm}
\label{spa-st-strong-algorithm-description}
\texttt{Algorithm SPA-ST-strong}, described in Algorithm \ref{algorithmSPA-STstrong}, begins by initialising an empty bipartite graph $G$ which will contain the provisional assignments of students to projects (and intuitively to lecturers). We remark that such assignments (i.e., edges in $G$) can subsequently be broken during the algorithm's execution. 

The \texttt{while} loop of the algorithm involves each student $s_i$ who is not adjacent to any project in $G$ and who has a non-empty list applying in turn to each project $p_j$ at the head of her list. Immediately, $s_i$ becomes provisionally assigned to $p_j$ in $G$ (and to $l_k$). If, by gaining a new provisional assignee, project $p_j$ becomes full or oversubscribed then we set $p_j$ as replete. Further, for each student $s_t$ in $\mathcal{L}_k^j$, such that $s_t$ is dominated in $\mathcal{L}_k^j$, we delete the pair $(s_t, p_j)$. As we will prove later, such pairs cannot belong to any strongly stable matching. Similarly, if by gaining a new provisional assignee, $l_k$ becomes full or oversubscribed then we set $l_k$ as replete. For each student $s_t$ in $\mathcal{L}_k$, such that $s_t$ is dominated in $\mathcal{L}_k$ and for each project $p_u \in P_k$ that $s_t$ finds acceptable, we delete the pair $(s_t, p_u)$. This continues until every student is provisionally assigned to one or more projects or has an empty list. At the point where the \texttt{while} loop terminates, we form the reduced assignment graph $G_r$ and we find the critical set $Z$ of students in $G_r$ (Lemma \ref{lemma:critical-set} describes how to find $Z$). As we will see later, no project $p_j \in \mathcal{N}(Z)$ can be assigned to any student in the tail of $\mathcal{L}_k^j$ in any strongly stable matching, so all such pairs are deleted.

At the termination of the inner \texttt{repeat-until} loop in line 21, i.e., when $Z$ is empty, if some project $p_j$ that is replete ends up undersubscribed, we let $s_r$ be any one of the most preferred students (according to $\mathcal{L}_k^j$) who was provisionally assigned to $p_j$ during some iteration of the algorithm but is not assigned to $p_j$ at this point (for convenience, we henceforth refer to such $s_r$ as the most preferred student rejected from $p_j$ according to $\mathcal{L}_k^j$). If the students at the tail of $\mathcal{L}_k$ (recalling that the tail of $\mathcal{L}_k$ is the least-preferred tie in $\mathcal{L}_k$ after any deletions might have occurred) are no better than $s_r$, it turns out that none of these students $s_t$ can be assigned to any project offered by $l_k$ in any strongly stable matching -- such pairs $(s_t, p_u)$, for each project $p_u \in P_k$ that $s_t$ finds acceptable, are deleted. The \texttt{repeat-until} loop is then potentially reactivated, and the entire process continues until every student is provisionally assigned to a project or has an empty list. 

At the termination of the outer \texttt{repeat-until} loop in line 30, if a student is adjacent in $G$ to a project $p_j$ via a bound edge, then we may potentially carry out extra deletions as follows. First, we let $l_k$ be the lecturer that offers $p_j$ and we let $U$ be the set of projects that are adjacent to $s_i$ in $G$ via an unbound edge. For each project $p_u \in U \setminus P_k$, it turns out that the pair $(s_i, p_u)$ cannot belong to any strongly stable matching, thus we delete all such pairs. Finally, we let $M$ be any feasible matching in the provisional assignment graph $G$. If $M$ is strongly stable relative to the given instance $I$ then $M$ is output as a strongly stable matching in $I$. Otherwise, the algorithm reports that no strongly stable matching exists in $I$.
We present \texttt{Algorithm SPA-ST-strong} in pseudocode form in Algorithm~\ref{algorithmSPA-STstrong}.

\begin{algorithm}[t!]

\caption{\texttt{Algorithm SPA-ST-strong}}
\begin{algorithmic}[1]
\Require {{\sc spa-st} instance $I$}
 
 \Ensure{a strongly stable matching in $I$ or ``no strongly stable matching exists in $I$''}
 
\State $G \gets \emptyset$
\Repeat{}
\Repeat{}
\While {some student $s_i$ is unassigned and has a non-empty list}
	\ForEach {project $p_j$ at the head of $s_i$'s list}
	\State $l_k \gets $ lecturer who offers $p_j$
	\State add the edge $(s_i, p_j)$ to $G$

		\If {$p_j$ is full or oversubscribed}
			\ForEach{student $s_t$ dominated in $\mathcal{L}_{k}^{j}$}
				\State delete $(s_t, p_j)$ 
			\EndFor
		\EndIf
		
	\If {$l_k$ is full or oversubscribed}	
		
			\ForEach{student $s_t$ dominated in $\mathcal{L}_{k}$}

				\ForEach{project $p_u \in P_k \cap A_t$ }
					 \State delete $(s_t, p_u)$ 
				\EndFor
			
			\EndFor
			
	\EndIf
		
	\EndFor

\EndWhile
\State form the reduced assignment graph $G_r$
\State find the critical set $Z$ of students
\ForEach{project $p_u \in \mathcal{N}(Z)$}
			\State $l_k \gets$ lecturer who offers $p_u$
				
				\ForEach{student $s_t$ at the tail of $\mathcal{L}_{k}^{u}$ }
				
					\State delete $(s_t, p_u)$ 
				\EndFor
			
			\EndFor		
\Until {$Z$ is empty}
\ForEach{$p_j \in \mathcal{P}$}
\If {$p_j$ is replete and $p_j$ is undersubscribed}
 \State $l_k \gets $ lecturer who offers $p_j$
\State $s_r \gets $ most preferred student rejected from $p_j$ in $\mathcal{L}_{k}^{j}$ \{any if $> 1$\}
%
\If{the students at the tail of $\mathcal{L}_k$ are no better than $s_r$}
				 	\ForEach{student $s_t$ at the tail of $\mathcal{L}_k$}

						\ForEach{project $p_u \in P_k \cap A_t$ }
							\State delete $(s_t, p_u)$ 

						\EndFor
			
					\EndFor
				
	\EndIf			

		\EndIf
\EndFor 
 \Until{every unassigned student has an empty list}
\ForEach{student $s_i$ in $G$ }
	\If{$s_i$ is adjacent in $G$ to a project $p_j$ via a bound edge}
	\State $l_k \gets $ lecturer who offers $p_j$
	\State $U$ $\gets$ unbound projects adjacent to $s_i$ in $G$
		\ForEach {$p_u \in U \setminus P_k$}
		\State delete $(s_i, p_u)$
		\EndFor
	\EndIf
\EndFor
\State $M \gets$ a feasible matching in $G$
\If {$M$ is a strongly stable matching in $I$}
\State \Return $M$

\Else
 \State \Return ``no strongly stable matching exists in $I$''

\EndIf
\end{algorithmic}
\label{algorithmSPA-STstrong} 
\end{algorithm}

\vspace{2mm}
\noindent
\textbf{Finding the critical set.} Consider the reduced assignment graph $G_r = (S_r, P_r, E_r)$ formed from $G$ at a given point during the algorithm's execution (at line 15). To find the critical set of students in $G_r$, first we need to construct a maximum matching $M_r$ in $G_r$, with respect to the revised quota $q_j^*$, for each $p_j \in P_r$. In this context, a \emph{matching} $M_r \subseteq E_r$ is such that $|M_r(s_i)| \leq 1$ for all $s_i \in S_r$, and $|M_r(p_j)| \leq q_j^*$ for all $p_j \in P_r$. We describe how to construct $M_r$ as follows:
\begin{itemize}
    \item [1.] Let $G_r'$ be the subgraph of $G_r$ induced by the dummy students adjacent to a project in $G_r$. First, find a maximum matching $M_r'$ in $G_r'$.
    \item [2.] Using $M_r'$ as an initial solution, find a maximum matching $M_r$ in $G_r$.\footnote{By making sure that all the dummy students are matched in step 1, we are guaranteed that no lecturer is oversubscribed with non-dummy students in $G_r$.}
\end{itemize}
If a student $s_i$ is not assigned to any project in $M_r$, we say that vertex $s_i$ is \emph{free} in $G_r$. Similarly, if a project $p_j$ is such that $p_j$ has fewer than $q_j^*$ assignees in $M_r$, we say that vertex $p_j$ is \emph{free} in $G_r$. An \emph{alternating path} in $G_r$ relative to $M_r$ is any simple path in which edges are alternately in, and not in, $M_r$. An \emph{augmenting path} in $G_r$ is an alternating path from a free student to a free project. The following lemmas are classical results with respect to matchings in bipartite graphs. 

\begin{restatable}[]{lemma}{max-card}
\label{lemma:max-card}
A matching $M_r$ in a reduced assignment graph $G_r$ has maximum cardinality if and only if there is no augmenting path relative to $M_r$ in $G_r$.
\end{restatable}

\begin{restatable}[]{lemma}{max-card-def}
\label{lemma:max-card-def}
Let $M_r$ be a maximum matching in the reduced assignment graph $G_r$. Then $|M_r| = |S_r| - \delta(G_r)$. 
\end{restatable}

The classical augmenting path algorithm can be used to obtain $M_r$, and as explained in \cite{IMS03}, this can be implemented to run in $O(\min\{n,\sum c_j\}m)$ time, where $n$ is the number of students and $m$ is the total length of the students' preference lists. Now that we have described how to construct a maximum matching in the reduced assignment graph, the following lemma tells us how to find the critical set of students.
\begin{restatable}[]{lemma}{criticalset}
\label{lemma:critical-set}
Given a maximum matching $M_r$ in the reduced assignment graph $G_r$, the critical set $Z$ consists of the set $U$ of unassigned students together with the set $U'$ of students reachable from a student in $U$ via an alternating path.
\end{restatable}
\begin{proof}
First, we note that $\delta(G_r) = |U|$. Let $C = U \cup U'$, we claim that $\delta(C) = |U|$. By the definition of $\delta(G_r)$, clearly  $\delta(C) \leq \delta(G_r) = |U|$. Now suppose that $\delta(C) < |U|$, then
\begin{eqnarray}
\label{ineq:critical-set}
|U| &>& \delta(C) \nonumber \\
&=& |U \cup U'| - \sum_{p_j \in \mathcal{N(C)}} q_j^* \nonumber \\
&=& |U| + |U'| - \sum_{p_j \in \mathcal{N(C)}} q_j^* \quad \quad \mbox{(since $U \cap U' = \emptyset$)} \nonumber \\
|U'| &<& \sum_{p_j \in \mathcal{N(C)}} q_j^* \enspace.
\end{eqnarray}
Inequality \ref{ineq:critical-set} implies that there is a project $p_j \in \mathcal{N}(C)$ such that $p_j$ has fewer than $q_j^*$ assignees in $M_r$; thus $p_j$ is free in $M_r$. We claim that every student that is assigned in $G_r$ to a project $p_{j'} \in \mathcal{N}(C)$ must be in $C$. For suppose there is a student $s_i \notin C$ such that $s_i$ is assigned to $p_{j'}$ in $M_r$. Since $p_{j'} \in \mathcal{N}(C)$, then $p_{j'}$ must be adjacent to some student $s_{i'} \in C$. Now, if $s_{i'} \in U$, then there is an alternating path from $s_{i'}$ to $s_{i}$ via $p_{j'}$, a contradiction. Otherwise, if $s_{i'} \in U'$, since $s_{i'}$ is reachable from a student in $U$ via an alternating path, $s_{i}$ is also reachable from the same student in $U$ via an alternating path. Hence our claim is established. Now, given that $p_j$ has fewer than $q_j^*$ assignees in $M_r$, since each student in $U'$ is reachable from a student in $U$ via an alternating path, and the students in $U \cup U'$ are collectively adjacent to projects in $\mathcal{N}(C)$, we can find an alternating path from a student in $U$ to $p_j$. Thus $M_r$ admits an augmenting path, contradicting the maximality of $M_r$.

Further, the critical set $Z$ must contain every student who is unassigned in some maximum matching in $G_r$. For, suppose not. Let $M_r^*$ be an arbitrary such matching (where $|M_r^*| = |S_r| - \delta(G_r)$), and suppose there is some student $s_i \in S_r \setminus Z$ such that $s_i$ is unassigned in $M_r^*$. There must be $\delta(G_r)$ unassigned students, with at most $\delta(G_r) - 1$ of these students contained in $Z$ (since $s_i \notin Z$). Hence $Z$ contains at least $|Z| - \delta(G_r) + 1$ assigned students. It follows that
$$\sum_{p_j \in \mathcal{N}(Z)} q_j^* \geq |Z| - \delta(G_r) + 1$$
or 
$$|Z| - \sum_{p_j \in \mathcal{N}(Z)} q_j^* \leq \delta(G_r) - 1$$
contradicting the required deficiency of $Z$.

But, for every $s_i \in U'$, there is a maximum matching in which $s_i$ is unassigned, obtainable from $M_r$ via an alternating path from a student in $U$ to $s_i$. Hence $C \subseteq Z$; and since $\delta(C) = \delta(Z)$, this completes the proof.
 \qed \end{proof}
 
\subsection{The non-triviality of extending \texttt{Algorithm HRT-strong} to {\sc spa-st}}
\label{spa-st-strong-difference}
\texttt{Algorithm SPA-ST-strong} is a non-trivial extension of \texttt{Algorithm HRT-strong} for {\sc hrt} \cite{IMS03}.  Here we outline the major distinctions between our algorithm and \texttt{Algorithm HRT-strong}, which indicate the challenges involved in extending the earlier approach to the {\sc spa-st} setting.

\begin{enumerate}
    \item Given a lecturer $l_k$, it is possible that during some iteration of our algorithm, some $p_j \in P_k$ is oversubscribed which causes $l_k$ to become full or oversubscribed (see Fig.~\ref{fig:strong-iteration(1)}(a) in Sect.~\ref{spa-st-strong-example-execution}, at the point where $s_3$ applies to $p_1$). Finding the dominated students in $\mathcal{L}_k$ becomes more complex in {\sc spa-st} -- to achieve this  we introduced the notion of quota (i.e., $q_k$) for $l_k$.
    \item To form $G_r$ in the {\sc spa-st} case, we extended the approach described in the {\sc hrt} case \cite{IMS03} by introducing the concept of lower rank edges for each lecturer who offers a project in $G_r$, and we also introduced dummy students.
    \item Lines 22 - 29 of \texttt{Algorithm SPA-ST-strong} refer to additional deletions that must be carried out in a certain situation; this type of deletion was also carried out in \texttt{Algorithm SPA-ST-super}  for super-stability \cite{OM18} (see the description corresponding to Fig.~\ref{fig:strong-iteration(2)} in Sect.~\ref{spa-st-strong-example-execution} for an example showing why we may need to carry out this type of deletion in the strong stability context).
    \item Constructing a feasible matching $M$ in $G$ in the {\sc spa-st} setting is much more challenging: we first identify some replete projects that must be full in $M$, denoted by $P_R^*$ (see the description corresponding to Fig.~\ref{fig:strong-iteration(3)} in Sect.~\ref{spa-st-strong-example-execution}). Also, in the {\sc hrt} case, when constructing $M$ from $G$, preference is given to a bound edge over an unbound edge; in general, this is not always true in the {\sc spa-st} case. 
\end{enumerate}

\subsection{Example algorithm execution}
 \label{spa-st-strong-example-execution}
In this section, we illustrate an  execution of \texttt{Algorithm SPA-ST-strong} with respect to the {\sc spa-st} instance $I_3$ shown in Fig.~\ref{fig:spa-st-instance-3}, which involves the set of students $\mathcal{S} = \{s_i: 1 \leq i \leq 8\}$, the set of projects $\mathcal{P} = \{p_j: 1 \leq j \leq 6\}$ and the set of lecturers $\mathcal{L} = \{l_k: 1 \leq k \leq 3\}$. The algorithm starts by initialising the bipartite graph $G = \{\}$, which will contain the provisional assignment of students to projects. We assume that the students become provisionally assigned to each project at the head of their list in subscript order. 
Figs.~\ref{fig:strong-iteration(1)}, \ref{fig:strong-iteration(2)} and \ref{fig:strong-iteration(3)} illustrate how this execution of \texttt{Algorithm SPA-ST-strong} proceeds with respect to $I_3$.

\begin{figure}[H]
\centering
\begin{tabular}{llll}
\hline
\texttt{Student preferences} & \qquad \qquad  & \texttt{Lecturer preferences} & offers\\ 
$s_1$: \;  $p_1$ \; $p_6$  &  & $\{3\}$ \; $l_1$: \; $s_8$ \; $s_7$ \;($s_1$ \; $s_2$ \; $s_3$) \;($s_4$ \; $s_5$) \; $s_6$ & $p_1$, $p_2$\\ 
$s_2$: \; $p_1$ \; $p_2$ &  & $\{2\}$ \; $l_2$: \;  $s_6$ \; $s_5$ \;($s_7$ \; $s_3$) &  $p_3$, $p_4$\\ 
$s_3$: \;($p_1$ \; $p_4$) &  & $\{3\}$ \; $l_3$: \;($s_1$ \; $s_4$) \; $s_8$ &  $p_5$, $p_6$\\
$s_4$: \; $p_2$ \;($p_5$ \; $p_6$)&  & & \\
$s_5$: \;($p_2$ \; $p_3$) &  & &\\ 
$s_6$: \;($p_2$ \; $p_4$) &  & &\\ 
$s_7$: \; $p_3$ \; $p_1$ &  & Project capacities: $c_1 = c_2 = c_6 = 2, \; c_3 = c_4 = c_5 = 1$& \\
$s_8$: \; $p_5$ \; $p_1$ &  & Lecturer capacities: $d_1 = d_3 = 3, \; d_2 = 2$&\\
\hline
\end{tabular}
\caption{\label{fig:spa-st-instance-3} \small An instance $I_3$ of {\sc spa-st}.}
\end{figure}
\begin{figure}[H]
        \centering
\subfigure[The provisional assignment graph $G^{(1)}$ at the end of the \texttt{while} loop, with the quota of each project labelled beside it.]{
            \begin{tikzpicture}[scale=0.525]
                \SetVertexNormal[MinSize = 5pt,LineWidth = 0.85pt]
              
                \Vertex[x=2, y=-1,LabelOut=true]{$p_1:2$}
                \Vertex[x=2, y=-2,LabelOut=true]{$p_2:2$}
                \Vertex[x=2, y=-3,LabelOut=true]{$p_3:1$}
                \Vertex[x=2, y=-4,LabelOut=true]{$p_4:1$}
                \Vertex[x=2, y=-5,LabelOut=true]{$p_5:1$}

				\SetVertexNormal[MinSize = 5pt,LineWidth = 0.85pt]
                \Vertex[x=-1, y=0,LabelOut=true,Lpos=180]{$s_1$} 
                \Vertex[x=-1, y=-1,LabelOut=true,Lpos=180]{$s_2$}
                \Vertex[x=-1, y=-2,LabelOut=true,Lpos=180]{$s_3$} 
                \Vertex[x=-1, y=-3,LabelOut=true,Lpos=180]{$s_4$} 
                \Vertex[x=-1, y=-4,LabelOut=true,Lpos=180]{$s_5$} 
                \Vertex[x=-1, y=-5,LabelOut=true,Lpos=180]{$s_6$} 
                \Vertex[x=-1, y=-6,LabelOut=true,Lpos=180]{$s_7$}        
                \Vertex[x=-1, y=-7,LabelOut=true,Lpos=180]{$s_8$}
                
                \draw [solid] (-0.85,0) to (1.85,-1); 
                \draw [solid] (-0.85,-1) to (1.85,-1);
                \draw [solid] (-0.85,-2) to (1.85,-1);
                \draw [solid] (-0.85,-3) to (1.85,-2);
                \draw [solid] (-0.85,-4) to (1.85,-2);
                \draw [solid] (-0.85,-4) to (1.85,-3);
                \draw [solid] (-0.85,-5) to (1.85,-4);
                \draw [solid] (-0.85,-6) to (1.85,-1);
                \draw [solid] (-0.85,-7) to (1.85,-5);
                 
            \end{tikzpicture}
\label{G(1)}} 
\qquad 
\subfigure[The reduced assignment graph $G_r^{(1)}$, with the revised quota of each project labelled beside it. The collection of the dashed edges is the maximum matching $M_r^{(1)}$.]{
         \begin{tikzpicture}[scale=0.525]
                \SetVertexNormal[MinSize = 5pt,LineWidth = 0.85pt]
               
                \Vertex[x=2, y=-1,LabelOut=true]{$p_1:1$}
                \Vertex[x=2, y=-2,LabelOut=true]{$p_2:2$}
			
				\SetVertexNormal[MinSize = 7pt,LineWidth = 0.85pt]
                \Vertex[x=-1, y=0,LabelOut=true,Lpos=180]{$s_1$} 
                \Vertex[x=-1, y=-1,LabelOut=true,Lpos=180]{$s_2$}
                \Vertex[x=-1, y=-2,LabelOut=true,Lpos=180]{$s_3$} 
                \Vertex[x=-1, y=-3,LabelOut=true,Lpos=180]{$s_4$}
                \Vertex[x=-1, y=-4,LabelOut=true,Lpos=180]{$s_{d_1}$}

                \draw [solid] (-0.85,0) to (1.85,-1); 
                \draw [dashed] (-0.85,-1) to (1.85,-1);
                \draw [solid] (-0.85,-2) to (1.85,-1);
                \draw [dashed] (-0.85,-3) to (1.85,-2);
                \draw [dashed] (-0.85,-4) to (1.85,-2);
            \end{tikzpicture}
\label{Gr(1)}}
\qquad 
\subfigure[The provisional assignment graph $G^{(1)}$ at the termination of iteration (1).]{
            \begin{tikzpicture}[scale=0.525]
                \SetVertexNormal[MinSize = 5pt,LineWidth = 0.85pt]
              
                \Vertex[x=2, y=-1,LabelOut=true]{$p_1:1$}
                \Vertex[x=2, y=-2,LabelOut=true]{$p_2:2$}
                \Vertex[x=2, y=-3,LabelOut=true]{$p_3:1$}
                \Vertex[x=2, y=-4,LabelOut=true]{$p_4:1$}
                \Vertex[x=2, y=-5,LabelOut=true]{$p_5:1$}

				\SetVertexNormal[MinSize = 5pt,LineWidth = 0.85pt]
                
                \Vertex[x=-1, y=-3,LabelOut=true,Lpos=180]{$s_4$} 
                \Vertex[x=-1, y=-4,LabelOut=true,Lpos=180]{$s_5$} 
                \Vertex[x=-1, y=-5,LabelOut=true,Lpos=180]{$s_6$} 
                \Vertex[x=-1, y=-6,LabelOut=true,Lpos=180]{$s_7$}        
                \Vertex[x=-1, y=-7,LabelOut=true,Lpos=180]{$s_8$}

                \draw [solid] (-0.85,-3) to (1.85,-2);
                \draw [solid] (-0.85,-4) to (1.85,-2);
                \draw [solid] (-0.85,-4) to (1.85,-3);
                \draw [solid] (-0.85,-5) to (1.85,-4);
                \draw [solid] (-0.85,-6) to (1.85,-1);
                \draw [solid] (-0.85,-7) to (1.85,-5);
                 
            \end{tikzpicture}
\label{G(1)end}} 
\caption{\small \label{fig:strong-iteration(1)} Iteration (1).}
\end{figure}
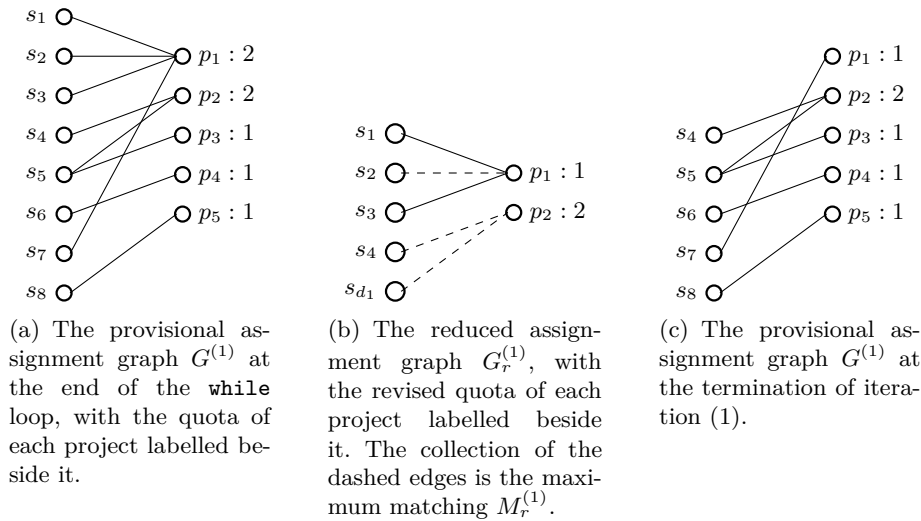
\noindent
\textit{Iteration 1:} At the termination of the \texttt{while} loop during the first iteration of the inner \texttt{repeat-until} loop, every student, except $s_3$, $s_6$ and $s_7$, is provisionally assigned to every project in the first tie on their preference list. Edge $(s_3, p_4) \notin G^{(1)}$ because $(s_3, p_4)$ was deleted as a result of $s_6$ becoming provisionally assigned to $p_4$, causing $s_3$ to be dominated in $\mathcal{L}_2^4$. Also, edge $(s_6, p_2) \notin G^{(1)}$ because $(s_6, p_2)$ was deleted as a result of $s_4$ becoming provisionally assigned to $p_2$, causing $s_6$ to be dominated in $\mathcal{L}_1$ (at that point in the algorithm, $\min \{d_G(l_1), \alpha_1\} = \min\{4,3\} = 3 = d_1$ and $s_6$ is worse than at least $d_1$ students who are provisionally assigned to $l_1$). Finally, edge $(s_7, p_3) \notin G^{(1)}$ because $(s_7, p_3)$ was deleted as a result of $s_5$ becoming provisionally assigned to $p_5$, causing $s_7$ to be dominated in $\mathcal{L}_2^3$. 
    
    To form $G_r^{(1)}$, the bound edges $(s_5, p_3), (s_6, p_4), (s_7, p_1)$ and $(s_8, p_5)$ are removed from the graph. We can verify that edges $(s_4, p_2)$ and $(s_5, p_2)$ are unbound, since they are lower rank edges for $l_1$. Also, since $p_1$ is oversubscribed, and each of $s_1, s_2$ and $s_3$ is at the tail of $\mathcal{L}_1^1$, edges $(s_1, p_1)$, $(s_2, p_1)$ and $(s_3, p_1)$ are unbound. Further, the revised quota of $l_1$ in $G_r^{(1)}$ is $2$, and the total revised quota of projects offered by $l_1$ (i.e., $p_1$ and $p_2$) is $3$. Thus, we add one dummy student vertex $s_{d_1}$ to $G_r^{1}$, and we add an edge between $s_{d_1}$ and $p_2$ (since $p_2$ is the only project  in $G_r^{(1)}$ adjacent to a student in the tail of $\mathcal{L}_1$ via a lower rank edge). With respect to the maximum matching $M_r^{(1)}$, it is clear that the critical set $Z^{(1)} =\{s_1, s_2, s_3\}$, thus we delete the edges $(s_1, p_1)$, $(s_2, p_1)$ and $(s_3, p_1)$ from $G^{(1)}$; and the inner \texttt{repeat-until} loop is reactivated.
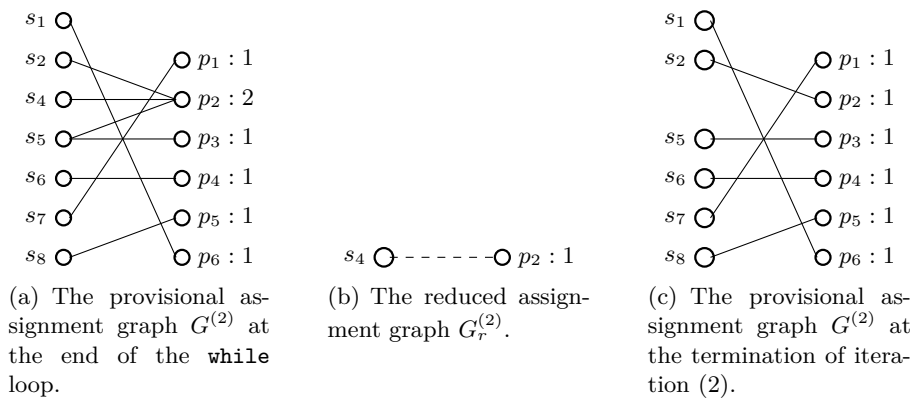
\begin{figure}[h]
        \centering
\subfigure[The provisional assignment graph $G^{(2)}$ at the end of the \texttt{while} loop.]{
            \begin{tikzpicture}[scale=0.525]
                \SetVertexNormal[MinSize = 5pt,LineWidth = 0.85pt]
              
                \Vertex[x=2, y=-1,LabelOut=true]{$p_1:1$}
                \Vertex[x=2, y=-2,LabelOut=true]{$p_2:2$}
                \Vertex[x=2, y=-3,LabelOut=true]{$p_3:1$}
                \Vertex[x=2, y=-4,LabelOut=true]{$p_4:1$}
                \Vertex[x=2, y=-5,LabelOut=true]{$p_5:1$}
                \Vertex[x=2, y=-6,LabelOut=true]{$p_6:1$}

				\SetVertexNormal[MinSize = 5pt,LineWidth = 0.85pt]
				\Vertex[x=-1, y=0,LabelOut=true,Lpos=180]{$s_1$}
                \Vertex[x=-1, y=-1,LabelOut=true,Lpos=180]{$s_2$}
                \Vertex[x=-1, y=-2,LabelOut=true,Lpos=180]{$s_4$} 
                \Vertex[x=-1, y=-3,LabelOut=true,Lpos=180]{$s_5$} 
                \Vertex[x=-1, y=-4,LabelOut=true,Lpos=180]{$s_6$} 
                \Vertex[x=-1, y=-5,LabelOut=true,Lpos=180]{$s_7$}        
                \Vertex[x=-1, y=-6,LabelOut=true,Lpos=180]{$s_8$}

                \draw [solid] (-0.85,0) to (1.85,-6);
                \draw [solid] (-0.85,-1) to (1.85,-2);
                \draw [solid] (-0.85,-2) to (1.85,-2);
                \draw [solid] (-0.85,-3) to (1.85,-2);
                \draw [solid] (-0.85,-3) to (1.85,-3);
                \draw [solid] (-0.85,-4) to (1.85,-4);
                \draw [solid] (-0.85,-5) to (1.85,-1);
                \draw [solid] (-0.85,-6) to (1.85,-5);
                 
            \end{tikzpicture}
\label{G(2)}} 
\qquad 
\subfigure[The reduced assignment graph $G_r^{(2)}$.]{
         \begin{tikzpicture}[scale=0.525]
                \SetVertexNormal[MinSize = 5pt,LineWidth = 0.85pt]

                \Vertex[x=2, y=0,LabelOut=true]{$p_2:1$}
			
				\SetVertexNormal[MinSize = 7pt,LineWidth = 0.85pt]

                \Vertex[x=-1, y=0,LabelOut=true,Lpos=180]{$s_4$}

                \draw [dashed] (-0.85,0) to (1.85,0);
            
            \end{tikzpicture}
\label{Gr(2)}}
\qquad
\subfigure[The provisional assignment graph $G^{(2)}$ at the termination of iteration (2).]{
            \begin{tikzpicture}[scale=0.525]
                \SetVertexNormal[MinSize = 5pt,LineWidth = 0.85pt]
              
                \Vertex[x=2, y=-1,LabelOut=true]{$p_1:1$}
                \Vertex[x=2, y=-2,LabelOut=true]{$p_2:1$}
                \Vertex[x=2, y=-3,LabelOut=true]{$p_3:1$}
                \Vertex[x=2, y=-4,LabelOut=true]{$p_4:1$}
                \Vertex[x=2, y=-5,LabelOut=true]{$p_5:1$}
                \Vertex[x=2, y=-6,LabelOut=true]{$p_6:1$}

				\SetVertexNormal[MinSize = 7pt,LineWidth = 0.85pt]
				\Vertex[x=-1, y=0,LabelOut=true,Lpos=180]{$s_1$}
                \Vertex[x=-1, y=-1,LabelOut=true,Lpos=180]{$s_2$}
                \Vertex[x=-1, y=-3,LabelOut=true,Lpos=180]{$s_5$} 
                \Vertex[x=-1, y=-4,LabelOut=true,Lpos=180]{$s_6$} 
                \Vertex[x=-1, y=-5,LabelOut=true,Lpos=180]{$s_7$}        
                \Vertex[x=-1, y=-6,LabelOut=true,Lpos=180]{$s_8$}
                
                \draw [solid] (-0.85,0) to (1.85,-6);
                \draw [solid] (-0.85,-1) to (1.85,-2);
                \draw [solid] (-0.85,-3) to (1.85,-3);
                \draw [solid] (-0.85,-4) to (1.85,-4);
                \draw [solid] (-0.85,-5) to (1.85,-1);
                \draw [solid] (-0.85,-6) to (1.85,-5);
                 
            \end{tikzpicture}
\label{G(2end)}} 
\caption{\small \label{fig:strong-iteration(2)} Iteration (2).}
\end{figure}

\noindent
\textit{Iteartion 2:} At the beginning of this iteration, each of $s_1$ and $s_2$ is unassigned and has a non-empty list; thus we add edges $(s_1, p_6)$ and $(s_2, p_2)$ to the provisional assignment graph obtained at the termination of iteration (1) to form $G_r^{(2)}$. It can be verified that every edge in $G_r^{(2)}$, except $(s_4, p_2)$ and $(s_5, p_2)$, is a bound edge. Clearly, the critical set $Z^{(2)} = \emptyset$, thus the inner \texttt{repeat-until} loop terminates. At this point, project $p_1$, which was replete during iteration (1), is undersubscribed in iteration (2). Moreover, the students at the tail of $\mathcal{L}_1$ (i.e., $s_4$ and $s_5$) are no better than $s_3$, where $s_3$ is one of the most preferred students rejected from $p_1$ according to $\mathcal{L}_1^1$; thus we delete edges $(s_4, p_2)$ and $(s_5, p_2)$. The outer \texttt{repeat-until} loop is then reactivated (since $s_4$ is unassigned and has a non-empty list).

\newpage
\begin{figure}[htbp]
        \centering
\subfigure[The provisional assignment graph $G^{(3)}$ at the end of the \texttt{while} loop.]{
            \begin{tikzpicture}[scale=0.525]
                \SetVertexNormal[MinSize = 7pt,LineWidth = 0.85pt]
              
                \Vertex[x=2, y=-1,LabelOut=true]{$p_1:2$}
                \Vertex[x=2, y=-2,LabelOut=true]{$p_2:1$}
                \Vertex[x=2, y=-3,LabelOut=true]{$p_3:1$}
                \Vertex[x=2, y=-4,LabelOut=true]{$p_4:1$}
                \Vertex[x=2, y=-5,LabelOut=true]{$p_5:1$}
                \Vertex[x=2, y=-6,LabelOut=true]{$p_6:2$}

				\SetVertexNormal[MinSize = 5pt,LineWidth = 0.85pt]
				\Vertex[x=-1, y=0,LabelOut=true,Lpos=180]{$s_1$}
                \Vertex[x=-1, y=-1,LabelOut=true,Lpos=180]{$s_2$}
                \Vertex[x=-1, y=-2,LabelOut=true,Lpos=180]{$s_4$} 
                \Vertex[x=-1, y=-3,LabelOut=true,Lpos=180]{$s_5$} 
                \Vertex[x=-1, y=-4,LabelOut=true,Lpos=180]{$s_6$} 
                \Vertex[x=-1, y=-5,LabelOut=true,Lpos=180]{$s_7$}        
                \Vertex[x=-1, y=-6,LabelOut=true,Lpos=180]{$s_8$}

                \draw [solid] (-0.85,0) to (1.85,-6);
                \draw [solid] (-0.85,-1) to (1.85,-2);
                
                \draw [solid] (-0.85,-2) to (1.85,-5);
                \draw [solid] (-0.85,-2) to (1.85,-6);
                \draw [solid] (-0.85,-3) to (1.85,-3);
                \draw [solid] (-0.85,-4) to (1.85,-4);
                \draw [solid] (-0.85,-5) to (1.85,-1);
                \draw [solid] (-0.85,-6) to (1.85,-1);
                 
            \end{tikzpicture}
\label{G(3)}} 
\caption{\small \label{fig:strong-iteration(3)} Iteration (3).}
\end{figure}
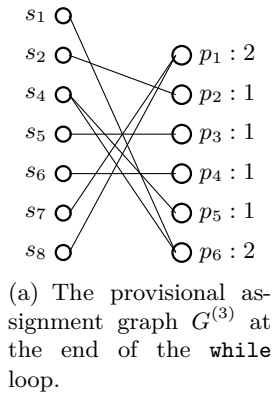
\noindent
\textit{Iteration 3:} At the beginning of this iteration, the only student that is unassigned and has a non-empty list is $s_4$; thus we add edges $(s_4, p_5)$ and $(s_4, p_6)$ to the provisional assignment graph obtained at the termination of iteration (2) to form $G_r^{(3)}$. The provisional assignment of $s_4$ to $p_5$ led to $p_5$ becoming oversubscribed; thus $(s_8, p_5)$ is deleted (since $s_8$ is dominated on $\mathcal{L}_3^5$). Further, $s_8$ becomes provisionally assigned to $p_1$. It can be verified that all the edges in $G_r^{(3)}$ are bound edges. Moreover, the reduced assignment graph $G_r^{(3)} = \emptyset$. 
    
Again, every unassigned students has an empty list. We also have that a project $p_2$, which was replete in iteration (2), is undersubscribed in iteration (3). However, no further deletion is carried out in line 29 of the algorithm, since the student at the tail of $\mathcal{L}_1$ (i.e., $s_2$) is better than $s_4$ and $s_5$, where $s_4$ and $s_5$ are the most preferred students rejected from $p_2$ according to $\mathcal{L}_1^2$. Hence, the \texttt{repeat-until} loop terminates. We observe that $P_R^* = \{p_5\}$, since $(s_8, p_5)$ has been deleted, $s_8$ prefers $p_5$ to her provisional assignment in $G$ and $l_3$ is undersubscribed. Thus we need to ensure $p_5$ fills up in the feasible matching $M$ constructed from $G$, so as to avoid $(s_8, p_5)$ from blocking $M$. Finally, the algorithm outputs the feasible matching $M = \{(s_1, p_6), (s_2, p_2), (s_4, p_5), (s_5, p_3), (s_6, p_4), (s_7, p_1), (s_8, p_1)\}$ as a strongly stable matching.

\subsection{Correctness of the algorithm}
\label{spa-st-strong-correctness-result}
We now present the following results regarding the correctness of \texttt{Algorithm SPA-ST-strong}. The first of these results deals with the fact that no strongly stable pair is ever deleted during the execution of the algorithm.
\begin{restatable}[]{lemma}{stspdeletion}
\label{stsp-deletion}
If a pair $(s_i, p_j)$ is deleted during the execution of \texttt{Algorithm SPA-ST-strong}, then $(s_i, p_j)$ does not belong to any strongly stable matching in $I$.
\end{restatable}

\noindent
In order to prove Lemma \ref{stsp-deletion}, we present Lemmas \ref{pair-deletion-inner}, \ref{pair-deletion-outer} and \ref{pair-deletion-new}.

\begin{lemma}
\label{pair-deletion-inner}
If a pair $(s_i, p_j)$ is deleted within the inner \texttt{repeat-until} loop during the execution of \texttt{Algorithm SPA-ST-strong}, then $(s_i, p_j)$ does not belong to any strongly stable matching in $I$.
\end{lemma}
\begin{proof}
Suppose that $(s_i, p_j)$ is the first strongly stable pair to be deleted within the inner \texttt{repeat-until} loop during an arbitrary execution $E$ of \texttt{Algorithm SPA-ST-strong}. Let $M^*$ be some strongly stable matching in which $s_i$ is assigned to $p_j$. Let $l_k$ be the lecturer who offers $p_j$. Suppose that $G$ is the provisional assignment graph immediately after the deletion of $(s_i, p_j)$. There are three cases to consider.

\begin{enumerate}[1.]
\item Suppose that $(s_i, p_j)$ is deleted (in line 10) because some other student became provisionally assigned to $p_j$ during $E$, causing $p_j$ to become full or oversubscribed, so that $s_i$ is dominated in $\mathcal{L}_k^j$. Since $(s_i, p_j) \in M^* \setminus G$, there is some student, say $s_r$, such that $l_k$ prefers $s_r$ to $s_i$ and $(s_r, p_j) \in G \setminus M^*$, for otherwise $p_j$ would be oversubscribed in $M^*$. We note that $s_r$ cannot be assigned to a project that she prefers to $p_j$ in any strongly stable matching, for otherwise some strongly stable pair must have been deleted before $(s_i, p_j)$, as $p_j$ must be in the head of $s_r$'s list when she applied. So $s_r$ is either unassigned in $M^*$, or $s_r$ prefers $p_j$ to $M^*(s_i)$ or is indifferent between them. Clearly for any combination of $l_k$ and $p_j$ being full or undersubscribed in $M^*$, it follows that $(s_r, p_j)$ blocks $M^*$, a contradiction.
\item Suppose that $(s_i, p_j)$ is deleted (in line 14) because some other student became provisionally assigned to a project offered by $l_k$ during $E$, causing $l_k$ to become full or oversubscribed, so that $s_i$ is dominated in $\mathcal{L}_k$. We denote by $C_k$ the set of projects that are full or oversubscribed in $G$, which are offered by $l_k$. We denote by $D_k$ the set of projects that are undersubscribed in $G$, which are offered by $l_k$. Clearly the projects offered by $l_k$ that are provisionally assigned to a student in $G$ at this point can be partitioned into $C_k$ and $D_k$. We consider two subcases.

	\begin{enumerate} [(i)]
\item Each student who is provisionally assigned in $G$ to a project in $D_k$ (if any) is also assigned to that same project in $M^*$. However, after the deletion of $(s_i, p_j)$, we know that 
\begin{eqnarray}
\sum\limits_{p_t \in C_k \cup D_k} q_t = \sum\limits_{p_t \in C_k} c_t + \sum\limits_{p_t \in D_k} d_G(p_t) \geq d_k;
\end{eqnarray}
i.e., the total quota of projects in $C_k \cup D_k$ is at least the capacity of $l_k$. Now, since $p_j$ has one more assignee in $M^*$ than it has provisional assignees in $G$, namely $s_i$, then some other project $p_{j'} \in C_k$ must have fewer than $c_{j'}$ assignees in $M^*$, for otherwise $l_k$ would be oversubscribed in $M^*$. This implies that there is some student, say $s_r$, such that $l_k$ prefers $s_r$ to $s_i$ and $(s_r, p_{j'}) \in G \setminus M^*$. Moreover, $s_{r}$ cannot be assigned to a project that she prefers to $p_{j'}$ in $M^*$, as explained in (1) above. Hence, $(s_{r}, p_{j'})$ blocks $M^*$, a contradiction.
\item Each project in $C_k$ at this point ends up full in $M^*$. This implies that there is some project $p_{j'} \in D_k$ with fewer assignees in $M^*$ than provisional assignees in $G$, for otherwise $l_k$ would be oversubscribed in $M^*$. Thus $p_{j'}$ is undersubscribed in $M^*$ (since $D_k$ is the set of undersusbcribed projects offered by $l_k$). Moreover, there is some student, say $s_r$, such that $l_k$ prefers $s_r$ to $s_i$ and $(s_r, p_{j'}) \in G \setminus M^*$. Following a similar argument in (i) above, $(s_r, p_{j'})$ blocks $M^*$, a contradiction.
	\end{enumerate}

\item Suppose that $(s_i, p_j)$ is deleted (in line 20) because $p_j$ is provisionally assigned to a student in the critical set $Z$ at some point, and at that point $s_i$ is in the tail of $\mathcal{L}_k^j$. We refer to the set of preference lists at that point as the current lists. Let $Z'$ be the set of students in $Z$ who are assigned in $M^*$ to a project from the head of their current lists, and let $P'$ be the set of projects in $\mathcal{N}(Z)$ assigned in $M^*$ to at least one student from the tail of its current list. We have that $p_j \in P'$, so $P' \neq \emptyset$. Consider $s_{i'} \in Z$. Now $s_{i'}$ cannot be assigned in $M^*$ to a project that she prefers to any project in the head of her current list, for otherwise some strongly stable pair must have been deleted before $(s_i, p_j)$. Hence, any student $s_{i'}$ in $Z$ who is provisionally assigned to $p_j$ must be in $Z'$, otherwise $(s_{i'}, p_j)$ would block $M^*$. Thus $Z' \neq \emptyset$. Also $Z \setminus Z' \neq \emptyset$, because $|Z| -  \sum \{q_j^*: p_{j} \in \mathcal{N}(Z)\} > 0$ and $|Z'| -  \sum \{q_j^*: p_{j} \in \mathcal{N}(Z)\} \leq |Z'| -  \sum \{q_j^*: p_{j} \in P'\} \leq 0$, since every student in $Z'$ is assigned in $M^*$ to a project in $P'$. 

We now claim that there must be an edge $(s_r, p_{j'})$ in $G_r$ such that $s_r \in Z \setminus Z'$ and $p_{j'} \in P'$. For otherwise $\mathcal{N}(Z \setminus Z') \subseteq \mathcal{N}(Z) \setminus P'$, and
 \begin{eqnarray*}
 |Z \setminus Z'| - \sum \{q_j^*: p_{j} \in \mathcal{N}(Z \setminus Z')\} & \geq &  |Z \setminus Z'| - \sum \{q_j^*: p_{j} \in \mathcal{N}(Z) \setminus P'\} \\
 & = & |Z| - \sum \{q_j^*: p_{j} \in \mathcal{N}(Z)\} - \left( |Z'| - \sum \{q_j^*: p_{j} \in P'\} \right) \\
 & \geq & |Z| - \sum \{q_j^*: p_{j} \in \mathcal{N}(Z)\},
 \end{eqnarray*}
since $|Z'| -  \sum \{q_j^*: p_{j} \in P'\} \leq 0$. Hence $Z \setminus Z'$ has deficiency at least that of $Z$, contradicting the fact that $Z$ is the critical set. Thus our claim is established, i.e., there is some student $s_r \in Z \setminus Z'$ and some project $p_{j'} \in P'$ such that $s_r$ is adjacent to $p_{j'}$ in $G_r$. Since $s_r$ is either unassigned in $M^*$ or prefers $p_{j'}$ to $M^*(s_r)$, and since the lecturer who offers $p_{j'}$ is indifferent between $s_r$ and at least one student in $M^*(p_{j'})$, we have that $(s_r, {p_{j'}})$ blocks $M^*$, a contradiction.
\end{enumerate}
\vspace{-0.3in}
\qed \end{proof}

\begin{lemma}
\label{pair-deletion-outer}
If a pair $(s_i, p_j)$ is deleted in line 29 during the execution of \texttt{Algorithm SPA-ST\-strong}, then $(s_i, p_j)$ does not belong to any strongly stable matching in $I$.
\end{lemma}
\begin{proof}
Suppose that $(s_i, p_j)$ was deleted in line 29 of the algorithm. Let $p_{j'}$ be some other project offered by $l_k$ which was replete during an iteration of the inner \texttt{repeat-until} loop and subsequently ends up undersubscribed at the end of the  loop, i.e., $p_{j'}$ plays the role of $p_j$ in line 23. Suppose that $s_{i'}$ plays the role of $s_r$ in line 25, i.e., $s_{i'}$ was the most preferred student rejected from $p_{j'}$ according to $\mathcal{L}_k^{j'}$ (possibly $s_{i'} = s_i$). Then $l_k$ prefers $s_{i'}$ to $s_i$ or is indifferent between them, since $s_i$ plays the role of $s_t$ at some for loop iteration in line 27. Moreover $s_{i'}$ was provisionally assigned to $p_{j'}$ during an iteration of the inner \texttt{repeat-until} loop but $(s_{i'}, p_{j'}) \notin G$ at the end of the loop. Thus $(s_{i'}, p_{j'}) \notin M^*$, since no strongly stable pair is deleted within the inner \texttt{repeat-until} loop, as proved in Lemma \ref{pair-deletion-inner}. 

\vspace{0.1in}
Let $l_{z_0} = l_k$, $p_{t_0} = p_{j'}$ and $s_{q_0} = s_{i'}$. Again, none of the students who are provisionally assigned to some project in $G$ can be assigned to any project better than their current assignment in any strongly stable matching as this would mean a strongly stable pair must have been deleted before $(s_i, p_j)$, as each student apply to projects in the head of her list. So, either (a) $s_{q_0}$ is unassigned in $M^*$ or $s_{q_0}$ prefers $p_{t_0}$ to $M^*(s_{q_0})$, or (b) $s_{q_0}$ is indifferent between $p_{t_0}$ and $M^*(s_{q_0})$. If (a) holds, then by the strong stability of $M^*$, $p_{t_0}$ is full in $M^*$ and $l_{z_0}$ prefers the worst student/s in $M^*(p_{t_0})$ to $s_{q_0}$; for if $p_{t_0}$ is undersubscribed in $M^*$ then $(s_{q_0}, p_{t_0})$ blocks $M^*$, since $s_i \in M^*(l_{z_0})$ and $l_{z_0}$ prefers $s_{q_0}$ to $s_i$ or is indifferent between them, a contradiction. If (b) holds we claim that $l_{z_0}$ prefers $s_{q_0}$ to $s_i$. Thus, irrespective of whether $M^*(s_{q_0})$ is offered by $l_{z_0}$ or not, the argument follows from (a), i.e., $p_{t_0}$ is full in $M^*$ and $l_{z_0}$ prefers the worst student/s in $M^*(p_{t_0})$ to $s_{q_0}$.

\vspace{0.1in}
Just before the deletion of $(s_i, p_j)$ occurred, $p_{t_0}$ is undersubscribed in $G$. Since $p_{t_0}$ is full in $M^*$, it follows that there exists some student, say $s_{q_1}$, such that $(s_{q_1}, p_{t_0}) \in M^* \setminus G$. We note that $l_{z_0}$ prefers $s_{q_1}$ to $s_{q_0}$. Let $p_{t_1} = p_{t_0}$. Since $(s_i, p_j)$ is the first strongly stable pair to be deleted, $s_{q_1}$ is provisionally assigned in $G$ to a project $p_{t_2}$ such that $s_{q_1}$ prefers $p_{t_2}$ to $p_{t_1}$. For otherwise, as students apply to projects in the head of their list, that would mean $(s_{q_1}, p_{t_1})$ must have been deleted during an iteration of the inner \texttt{repeat-until} loop, a contradiction. We note that $p_{t_2} \neq p_{t_1}$, since $(s_{q_1} , p_{t_2}) \in G$ and $(s_{q_1} , p_{t_1}) \notin G$. Let $l_{z_1}$ be the lecturer who offers $p_{t_2}$. By the strong stability of $M^*$, it follows that either 
\begin{itemize}
\item[(i)] $p_{t_2}$ is full in $M^*$ and $l_{z_1}$ prefers the worst student/s in $M^*(p_{t_2})$ to $s_{q_1}$, or
\item[(ii)] $p_{t_2}$ is undersubscribed in $M^*$, $l_{z_1}$ is full in $M^*$, $s_{q_1} \notin M^*(l_{z_1})$ and $l_{z_1}$ prefer the worst student/s in $M^*(l_{z_1})$ to $s_{q_1}$.
\end{itemize}

\vspace{0.1in}
Otherwise $(s_{q_1}, p_{t_2})$ blocks $M^*$. In case (i), there exists some student $s_{q_2} \in M^*(p_{t_2}) \setminus G(p_{t_2})$. Let $p_{t_3} = p_{t_2}$. In case (ii), there exists some student $s_{q_2} \in M^*(l_{z_1}) \setminus G(l_{z_1})$. We note that $l_{z_1}$ prefers $s_{q_2}$ to $s_{q_1}$. Now, suppose $M^*(s_{q_2}) = p_{t_3}$ (possibly $p_{t_3} = p_{t_2}$). It is clear that $s_{q_2} \neq s_{q_1}$. Applying similar reasoning as for $s_{q_1}$, $s_{q_2}$ is assigned in $G$ to a project $p_{t_4}$ such that $s_{q_2}$ prefers $p_{t_4}$ to $p_{t_3}$. Let $l_{z_2}$ be the lecturer who offers $p_{t_4}$. We are identifying a sequence $\langle s_{q_i}\rangle_{i \geq 0}$ of students, a sequence $\langle p_{t_i}\rangle_{i \geq 0}$ of projects, and a sequence $\langle l_{z_i}\rangle_{i \geq 0}$ of lecturers, such that, for each $i \geq 1$

\begin{enumerate}
\item $s_{q_{i}}$ prefers $p_{t_{2i}}$ to $p_{t_{2i-1}}$,
\item $(s_{q_i}, p_{t_{2i}}) \in G$ and $(s_{q_i}, p_{t_{2i - 1}}) \in M^*$,
\item $l_{z_i}$ prefers $s_{q_{i+1}}$ to $s_{q_{i}}$; also, $l_{z_i}$ offers both $p_{t_{2i}}$ and $p_{t_{2i+1}}$ (possibly $p_{t_{2i}} = p_{t_{2i+1}}$).
\end{enumerate}

\vspace{0.1in}
First we claim that for each new project that we identify, $p_{t_{2i}} \neq p_{t_{2i-1}}$ for $i \geq 1$. Suppose $p_{t_{2i}} = p_{t_{2i-1}}$ for some $i \geq 1$. From above $s_{q_{i}}$ was identified by $l_{z_{i-1}}$ such that $(s_{q_{i}}, p_{t_{2i-1}}) \in M^* \setminus G$. Moreover $(s_{q_{i}}, p_{t_{2i}}) \in G$. Hence we reach a contradiction. Clearly, for each student $s_{q_i}$ for $i \geq 1$ we identify, $s_{q_i}$ must be assigned to distinct projects in $G$ and in $M^*$.

\vspace{0.1in}
Next we claim that for each new student $s_{q_i}$ that we identify, $s_{q_i} \neq s_{q_t}$ for $1 \leq t < i$. We prove this by induction on $i$. For the base case, clearly $s_{q_2} \neq s_{q_1}$. We assume that the claim holds for some $i \geq 1$, i.e., the sequence $s_{q_{1}}, s_{q_2}, \ldots, s_{q_{i}}$ consists of distinct students. We show that the claim holds for $i+1$, i.e, the sequence $s_{q_{1}}, s_{q_2}, \ldots, s_{q_{i}}, s_{q_{i+1}}$ also consists of distinct students. Clearly $s_{q_{i+1}} \neq s_{q_{i}}$ since $l_{z_{i}}$ prefers $s_{q_{i+1}}$ to $s_{q_{i}}$. Thus, it suffices to show that $s_{q_{i+1}} \neq s_{q_{j}}$ for $1 \leq j \leq i-1$. Now, suppose $s_{q_{i+1}} = s_{q_{j}}$ for $1 \leq j \leq i-1$. This implies that $s_{q_{j}}$ was identified by $l_{z_{i}}$ and clearly $l_{z_{i}}$ prefers $s_{q_{j}}$ to $s_{q_{j-1}}$. Now since $s_{q_{i+1}}$ was also identified by $l_{z_{i}}$ to avoid the blocking pair $(s_{q_i}, p_{t_{2_i}})$ in $M^*$, it follows that either
(i) $p_{t_{2i}}$ is full in $M^*$, or
(ii) $p_{t_{2i}}$ is undersubscribed in $M^*$ and $l_{z_{i}}$ is full in $M^*$. We consider each cases further as follows.
\begin{itemize}
\item[(i)] If $p_{t_{2i}}$ is full in $M^*$, we know that $(s_{q_{i}}, p_{t_{2i}}) \in G \setminus M^*$. Moreover $s_{q_j}$ was identified by $l_{z_{i+1}}$ because of case (i). Furthermore $(s_{q_{j-1}}, p_{t_{2i}}) \in G \setminus M^*$. In this case, $p_{t_{2i+1}} = p_{t_{2i}}$ and we have that
$$(s_{q_{i}}, p_{t_{2i+1}})\in G \setminus M^* \mbox{ and } (s_{q_{i+1}}, p_{t_{2i+1}}) \in M^* \setminus G,$$ 
$$(s_{q_{j-1}}, p_{t_{2i+1}}) \in G \setminus M^* \mbox{ and } (s_{q_{j}}, p_{t_{2i+1}}) \in M^* \setminus G.$$
By the inductive hypothesis, the sequence $s_{q_{1}}, s_{q_2}, \ldots, s_{q_{j-1}}, $ $s_{q_j}, \ldots, s_{q_{i}}$ consists of distinct students. This implies that $s_{q_{i}} \neq s_{q_{j-1}}$. Thus since $p_{t_{2i+1}}$ is full in $M^*$, $l_{z_{i}}$ should have been able to identify distinct students $s_{q_j}$ and $s_{q_{i+1}}$ to avoid the blocking pairs $(s_{q_{j-1}}, p_{t_{2i+1}})$ and $(s_{q_{i}}, p_{t_{2i+1}})$ respectively in $M^*$, a contradiction.
\item[(ii)] $p_{t_{2i}}$ is undersubscribed in $M^*$ and $l_{z_{i}}$ is full in $M^*$. Similarly as in case (i) above, we have that
$$s_{q_{i}} \in G(l_{z_i}) \setminus M^*(l_{z_i}) \mbox{ and } s_{q_{i+1}} \in M^*(l_{z_i}) \setminus G(l_{z_i}),$$ 
$$s_{q_{j-1}} \in G(l_{z_i}) \setminus M^*(l_{z_i}) \mbox{ and } s_{q_{j}} \in M^*(l_{z_i}) \setminus G(l_{z_i}).$$
Since $s_{q_{i}} \neq s_{q_{j-1}}$ and $l_{z_{i}}$ is full in $M^*$, $l_{z_{i}}$ should have been able to identify distinct students $s_{q_j}$ and $s_{q_{i+1}}$ corresponding to students $s_{q_{j-1}}$ and $s_{q_{i}}$ respectively, a contradiction.
\end{itemize}  
This completes the induction step. As the sequence of distinct students and projects is infinite, we reach an immediate contradiction.
\qed \end{proof}

\begin{lemma}
\label{pair-deletion-new}
If a pair $(s_i, p_j)$ is deleted in line 36 during the execution of \texttt{Algorithm SPA-ST\-strong}, then $(s_i, p_j)$ does not belong to any strongly stable matching in $I$.
\end{lemma}
\begin{proof}
Suppose that $(s_i, p_j)$ is the first strongly stable pair to be deleted during an arbitrary execution $E$ of \texttt{Algorithm SPA-ST-strong}. Then by Lemmas \ref{pair-deletion-inner} and \ref{pair-deletion-outer}, $(s_i, p_j)$ was deleted in line 36. Let $l_k$ be the lecturer who offers $p_j$ and let $M^*$ be some strongly stable matching in which $s_i$ is assigned to $p_j$. Let $G$ be the provisional assignment graph. At this point in the algorithm where the deletion of $(s_i, p_j)$ occured, $s_i$ is adjacent to some other project $p_{j'} \in G$ via a bound edge, where $p_{j'}$ is offered by $l_{k'}$ (note that $l_k \neq l_{k'}$). By the definition of a bound edge, it follows that either $p_{j'}$ is not oversubscribed in $G$ or $l_{k'}$ prefers $s_i$ to some student in $G(p_{j'})$ (or both). Also, either $(s_i, p_{j'})$ is not a lower rank edge or $l_{k'}$ prefers $s_i$ to some student in $G(l_{k'})$ (or both). 

Let $l_{z_0} = l_{k'}$, $p_{t_0} = p_{j'}$ and $s_{q_0} = s_{i}$. We note that $s_{q_0}$ is indifferent between $p_j$ and $p_{t_0}$. Now, since $(s_{q_0}, p_{t_0}) \in G \setminus M^*$, by the strong stabiity of $M^*$, either (i) or (ii) holds as follows:
\begin{enumerate}[(i)]
\item $p_{t_0}$ is full in $M^*$ and $l_{z_0}$ prefers the worst student/s in $M^*(p_{t_0})$ to $s_{q_0}$ or is indifferent between them;
\item $p_{t_0}$ is undersubscribed in $M^*$, $l_{z_0}$ is full in $M^*$ and $l_{z_0}$ prefers the worst students in $M^*(l_{z_0})$ to $s_{q_0}$ or is indifferent between them.
\end{enumerate}
Otherwise $(s_{q_0}, p_{t_0})$ blocks $M^*$. In case (i), since $s_{q_0}$ is bound to $p_{t_0}$ in $G$ and since $(s_{q_0}, p_{t_0}) \in G \setminus M^*$, there exists some student $s_{q_1} \in M^*(p_{t_0}) \setminus G(p_{t_0})$.   Let $p_{t_1} = p_{t_0}$. In case (ii), since $s_{q_0}$ is assigned in $G$ to a project offered by $l_{z_0}$ (i.e., $p_{t_0}$) via a bound edge, and since $s_{q_0}$ is not assigned to $l_{z_0}$ in $M^*$, there exists some student $s_{q_1} \in M^*(l_{z_0}) \setminus G(l_{z_0})$. We note that $l_{z_0}$ either prefers $s_{q_1}$ to $s_{q_0}$ or is indifferent between them; clearly $s_{q_1} \neq s_{q_0}$.

Now, suppose $M^*(s_{q_1}) = p_{t_1}$ (possibly $p_{t_1} = p_{t_0}$). Since $(s_{q_0}, p_j)$ is the first strongly stable pair to be deleted, $s_{q_1}$ is provisionally assigned in $G$ to a project $p_{t_2}$ such that $s_{q_1}$ prefers $p_{t_2}$ to $p_{t_1}$. For otherwise, as students apply to projects in the head of their list and since $(s_{q_1}, p_{t_1}) \notin G$, that would mean $(s_{q_1}, p_{t_1})$ must have been deleted during an iteration of the \texttt{repeat-until} loop, a contradiction. We note that $p_{t_2} \neq p_{t_1}$, since $(s_{q_1} , p_{t_2}) \in G$ and $(s_{q_1} , p_{t_1}) \notin G$. Let $l_{z_1}$ be the lecturer who offers $p_{t_2}$. Again by the strong stability of $M^*$, either (i) or (ii) holds as follows:
 \begin{enumerate}[(i)]
\item $p_{t_2}$ is full in $M^*$ and $l_{z_1}$ prefers the worst student/s in $M^*(p_{t_2})$ to $s_{q_1}$;
\item $p_{t_2}$ is undersubscribed in $M^*$, $l_{z_1}$ is full in $M^*$ and $l_{z_1}$ prefers the worst students in $M^*(l_{z_1})$ to $s_{q_1}$.
\end{enumerate}

Otherwise $(s_{q_1}, p_{t_2})$ blocks $M^*$. In case (i), there exists some student $s_{q_2} \in M^*(p_{t_2}) \setminus G(p_{t_2})$. Let $p_{t_3} = p_{t_2}$. In case (ii), there exists some student $s_{q_2} \in M^*(l_{z_1}) \setminus G(l_{z_1})$. We note that $l_{z_1}$ prefers $s_{q_2}$ to $s_{q_1}$; again it is clear that $s_{q_2} \neq s_{q_1}$. Now, suppose $M^*(s_{q_2}) = p_{t_3}$ (possibly $p_{t_3} = p_{t_2}$). Applying similar reasoning as for $s_{q_1}$, $s_{q_2}$ is provisionally assigned in $G$ to a project $p_{t_4}$ such that $s_{q_2}$ prefers $p_{t_4}$ to $p_{t_3}$. Let $l_{z_2}$ be the lecturer who offers $p_{t_4}$. We are identifying a sequence $\langle s_{q_i}\rangle_{i \geq 0}$ of students, a sequence $\langle p_{t_i}\rangle_{i \geq 0}$ of projects, and a sequence $\langle l_{z_i}\rangle_{i \geq 0}$ of lecturers, such that, for each $i \geq 1$
\begin{enumerate}
\item $s_{q_{i}}$ prefers $p_{t_{2i}}$ to $p_{t_{2i-1}}$,
\item $(s_{q_i}, p_{t_{2i}}) \in G$ and $(s_{q_i}, p_{t_{2i - 1}}) \in M^*$,
\item $l_{z_i}$ prefers $s_{q_{i+1}}$ to $s_{q_{i}}$; also, $l_{z_i}$ offers both $p_{t_{2i}}$ and $p_{t_{2i+1}}$ (possibly $p_{t_{2i}} = p_{t_{2i+1}}$).
\end{enumerate}

Following a similar argument as in the proof of Lemma~\ref{pair-deletion-outer}, we can identify an infinite sequence of distinct students and projects, a contradiction.
\qed \end{proof}

Lemmas \ref{pair-deletion-inner}, \ref{pair-deletion-outer} and \ref{pair-deletion-new} immediately give rise to Lemma \ref{stsp-deletion}. The next two lemmas will be used as a tool in the proof of the remaining lemmas.

\begin{restatable}[]{lemma}{studentlemma}
\label{student}
Every student who is assigned to a project (intuitively a lecturer) in the final provisional assignment graph $G$ must be assigned in any feasible matching $M$.
\end{restatable}
\begin{proof}
Suppose $s_i$ is a student who is provisionally assigned to some project in $G$, but $s_i$ is not assigned in a feasible matching $M$, then $s_i$ must be in the critical set $Z$, and hence $Z \neq \emptyset$, a contradiction.
\qed \end{proof}

\begin{restatable}[]{lemma}{lecturerundersubscribedtool}
\label{lemma:lecturer-undersubscribed-tool}
Let $M$ be a feasible matching in the final provisional assignment graph $G$ and let $M^*$ be any strongly stable matching. Let $l_k$ be an arbitrary lecturer; (i) if $l_k$ is undersubscribed in $M^*$ then every student who is assigned to $l_k$ in $M$ is also assigned to $l_k$ in $M^*$, and (ii) if $l_k$ is undersubscribed in $M$ then $l_k$ has the same number of assignees in $M^*$ as in $M$. 
\end{restatable}
\begin{proof}
Let $l_k$ be an arbitrary lecturer. First, we show that (i) holds. Suppose otherwise, then there exists a student, say $s_i$, such that $s_i \in M(l_k) \setminus M^*(l_k)$. Moreover, there exists some project $p_j \in P_k$ such that $s_i \in M(p_j) \setminus M^*(p_j)$. By Lemma \ref{stsp-deletion}, $s_i$ cannot be assigned to a project that she prefers to $p_j$ in $M^*$. Also, by the strong stability of $M^*$, $p_j$ is full in $M^*$ and $l_k$ prefers the worst student/s in $M^*(p_j)$ to $s_i$. 

Let $l_{z_0} = l_k$, $p_{t_0} = p_{j}$, and $s_{q_0} = s_{i}$.
As $p_{t_0}$ is full in $M^*$ and no project is oversubscribed in $M$, there exists some student $s_{q_1} \in M^*(p_{t_0}) \setminus M(p_{t_0})$ such that $l_{z_0}$ prefers $s_{q_1}$ to $s_{q_0}$. Let $p_{t_1} = p_{t_0}$. By Lemma \ref{stsp-deletion}, $s_{q_1}$ is assigned in $M$ to a project $p_{t_2}$ such that $s_{q_1}$ prefers $p_{t_2}$ to $p_{t_1}$. 
We note that $s_{q_1}$ cannot be indifferent between $p_{t_2}$ and $p_{t_1}$; for otherwise, as students apply to projects in the head of their list, since $(s_{q_1}, p_{t_1}) \notin M$, that would mean $(s_{q_1}, p_{t_1})$ must have been deleted during the algorithm's execution, contradicting Lemma \ref{stsp-deletion}. It follows that $s_{q_1} \in M(p_{t_2}) \setminus M^*(p_{t_2})$. Let $l_{z_1}$ be the lecturer who offers $p_{t_2}$. By the strong stability of $M^*$, either

\begin{itemize}
\item[(i)] $p_{t_2}$ is full in $M^*$ and $l_{z_1}$ prefers the worst student/s in $M^*(p_{t_2})$ to $s_{q_1}$, or
\item[(ii)] $p_{t_2}$ is undersubscribed in $M^*$, $l_{z_1}$ is full in $M^*$, $s_{q_1} \notin M^*(l_{z_1})$ and $l_{z_1}$ prefers the worst student/s in $M^*(l_{z_1})$ to $s_{q_1}$.
\end{itemize}

Otherwise $(s_{q_1}, p_{t_2})$ blocks $M^*$. In case (i), there exists some student $s_{q_2} \in M^*(p_{t_2}) \setminus M(p_{t_2})$. Let $p_{t_3} = p_{t_2}$. In case (ii), there exists some student $s_{q_2} \in M^*(l_{z_1}) \setminus M(l_{z_1})$. We note that $l_{z_1}$ prefers $s_{q_2}$ to $s_{q_1}$. Now, suppose $M^*(s_{q_2}) = p_{t_3}$ (possibly $p_{t_3} = p_{t_2}$). It is clear that $s_{q_2} \neq s_{q_1}$. Applying similar reasoning as for $s_{q_1}$, student $s_{q_2}$ is assigned in $M$ to a project $p_{t_4}$ such that $s_{q_2}$ prefers $p_{t_4}$ to $p_{t_3}$. Let $l_{z_2}$ be the lecturer who offers $p_{t_4}$. We are identifying a sequence $\langle s_{q_i}\rangle_{i \geq 0}$ of students, a sequence $\langle p_{t_i}\rangle_{i \geq 0}$ of projects, and a sequence $\langle l_{z_i}\rangle_{i \geq 0}$ of lecturers, such that, for each $i \geq 1$

\begin{enumerate}
\item $s_{q_{i}}$ prefers $p_{t_{2i}}$ to $p_{t_{2i-1}}$,
\item $(s_{q_i}, p_{t_{2i}}) \in G$ and $(s_{q_i}, p_{t_{2i - 1}}) \in M^*$,
\item $l_{z_i}$ prefers $s_{q_{i+1}}$ to $s_{q_{i}}$; also, $l_{z_i}$ offers both $p_{t_{2i}}$ and $p_{t_{2i+1}}$ (possibly $p_{t_{2i}} = p_{t_{2i+1}}$).
\end{enumerate}

Following a similar argument as in the proof of Lemma~\ref{pair-deletion-outer}, we can identify an infinite sequence of distinct students and projects, a contradiction. Hence, if $l_k$ is undersubscribed in $M^*$ then every student who is assigned to $l_k$ in $M$ is also assigned to $l_k$ in $M^*$. 

Next, we show that (ii) holds. By the first claim, any lecturer who is full in $M$ is also full in $M^*$, and any lecturer who is undersubscribed in $M$ has as many assignees in $M^*$ as she has in $M$. Hence 
\begin{eqnarray}
\label{ineq:undersubscribed-lecturer-1}
\sum_{l_k \in \mathcal{L}}{|M(l_k)|} \leq \sum_{l_k \in \mathcal{L}}{|M^*(l_k)|} \enspace.
\end{eqnarray}
Let $S_1$ and $S_2$ denote the set of students who are assigned to a project in $M$ and $M^*$ respectively. By Lemma \ref{student}, each student who is provisionally assigned to a project in the final provisional assignment graph $G$ must be assigned to a project in $M$. Moreover, any student who is not provisionally assigned to a project in $G$ must have an empty list. It follows that these students are unassigned in $M^*$, since Lemma \ref{stsp-deletion} guarantees that no strongly stable pairs are deleted. Thus $|S_2| \leq |S_1|$. Further, we have that
\begin{eqnarray}
\label{ineq:undersubscribed-lecturer-2}
\sum_{l_k \in \mathcal{L}}{|M^*(l_k)|} = |S_2| \leq |S_1| =\sum_{l_k \in \mathcal{L}}{|M(l_k)|},
\end{eqnarray}
From Inequality \ref{ineq:undersubscribed-lecturer-1} and \ref{ineq:undersubscribed-lecturer-2}, it follows that $|M(l_k)| = |M^*(l_k)|$ for each $l_k \in \mathcal{L}$.
 \qed \end{proof}
 
The next three lemmas deal with the case that \texttt{Algorithm SPA-ST-strong} reports the non-existence of a strongly stable matching in $I$.

\begin{restatable}[]{lemma}{repletelecturer}
\label{replete-lecturer}
Let $M$ be a feasible matching in the final provisional assignment graph $G$. Suppose that (a) some non-replete lecturer $l_k$ has fewer assignees in $M$ than provisional assignees in $G$, or (b) some replete lecturer is not full in $M$. Then $I$ admits no strongly stable matching.
\end{restatable}
 \begin{proof}
Suppose that $M^*$ is a strongly stable matching for the instance. By Lemma \ref{student}, each student who is provisionally assigned to a project in $G$ must be assigned to a project in $M$. Moreover, any student who is not provisionally assigned to a project in $G$ must have an empty list. It follows that these students are unassigned in $M^*$, since Lemma \ref{stsp-deletion} guarantees that no strongly stable pairs are deleted. Thus $|M^*| \leq |M|$.

Suppose that condition (a) is satisfied. Then some non-replete lecturer $l_{k'}$ satisfies $|M(l_{k'})| < d_G(l_{k'})$, where $d_G(l_{k'})$ is the number of students provisionally assigned in $G$ to a project offered by $l_{k'}$. As $l_{k'}$ is non-replete, it follows that $d_G(l_{k'}) < d_{k'}$. Now $|M(l_{k})| \leq \min \{d_k, d_G(l_{k})\}$ for all $l_k \in \mathcal{L}$. Hence
\begin{align}
\label{ineq:non-replete-lecturer}
|M| = \sum_{l_k \in \mathcal{L}} |M(l_k)| < \sum_{l_k \in \mathcal{L}} \min \{d_k, d_G(l_{k})\}\enspace.
\end{align}

Now, suppose that $|M^*(l_{k})| \geq \min \{d_k, d_G(l_{k})\}$ for all $l_k \in \mathcal{L}$. Then $|M^*| > |M|$ by \ref{ineq:non-replete-lecturer}, a contradiction. Hence $|M^*(l_{k})| < \min \{d_k, d_G(l_{k})\}$ for some $l_{k} \in \mathcal{L}$. This implies that $l_{k}$ is undersubcribed in $M^*$. Moreover, $l_k$ has fewer assignees in $M^*$ than provisional assignees in $G$. Thus there exists some student $s_i$ who is provisionally assigned to $l_k$ in $G$ but not in $M^*$. It follows that there exists some project $p_j \in P_k$ such that $(s_i, p_j) \in G \setminus M^*$. By Lemma \ref{stsp-deletion}, $s_i$ is not assigned to a project that she prefers to $p_j$ in $M^*$. Also, by the strong stability of $M^*$, $p_j$ is full in $M^*$ and $l_k$ prefers the worst student/s in $M^*(p_j)$ to $s_i$; for if $p_j$ is undersubscribed in $M^*$ then $(s_i, p_j)$ blocks $M^*$, a contradiction. Since $p_j$ is full in $M^*$ with students that are better than $s_i$ in $\mathcal{L}_k^j$ and $(s_i, p_j) \in G \setminus M^*$, then there is some student, say $s_{i'}$, such that $l_k$ prefers $s_{i'}$ to $s_i$ and $(s_{i'}, p_j) \in M^* \setminus G$. For if all the students assigned to $p_j$ in $M^*$ are also assigned to $p_j$ in $G$, then $s_i$ would be dominated in $\mathcal{L}_k^j$ and thus $(s_i, p_j)$ would have been deleted in $G$. 

Let $l_{z_0} = l_k$, $p_{t_0} = p_{t_1} = p_{j}$, $s_{q_0} = s_{i}$, $s_{q_1} = s_{i'}$. Again, by Lemma \ref{stsp-deletion}, $s_{q_1}$ is provisionally assigned in $G$ to a project $p_{t_2}$ such that $s_{q_1}$ prefers $p_{t_2}$ to $p_{t_1}$. For otherwise, as students apply to projects in the head of their list, since $(s_{q_1}, p_{t_1}) \notin G$, that would mean $(s_{q_1}, p_{t_1})$ must have been deleted during the algorithm's execution, a contradiction. We note that $p_{t_2} \neq p_{t_1}$, since $(s_{q_1}, p_{t_2}) \in G$ and $(s_{q_1}, p_{t_1}) \notin G$. Let $l_{z_1}$ be the lecturer who offers $p_{t_2}$. By the strong stability of $M^*$, either

\begin{itemize}
\item[(i)] $p_{t_2}$ is full in $M^*$ and $l_{z_1}$ prefers the worst student/s in $M^*(p_{t_2})$ to $s_{q_1}$, or
\item[(ii)] $p_{t_2}$ is undersubscribed in $M^*$, $l_{z_1}$ is full in $M^*$, $s_{q_1} \notin M^*(l_{z_1})$ and $l_{z_1}$ prefers the worst student/s in $M^*(l_{z_1})$ to $s_{q_1}$.
\end{itemize}

Otherwise $(s_{q_1}, p_{t_2})$ blocks $M^*$. In case (i), there exists some student $s_{q_2} \in M^*(p_{t_2}) \setminus G(p_{t_2})$. Let $p_{t_3} = p_{t_2}$. In case (ii), there exists some student $s_{q_2} \in M^*(l_{z_1}) \setminus G(l_{z_1})$. We note that $l_{z_1}$ prefers $s_{q_2}$ to $s_{q_1}$. Now, suppose $M^*(s_{q_2}) = p_{t_3}$ (possibly $p_{t_3} = p_{t_2}$). It is clear that $s_{q_2} \neq s_{q_1}$. Applying similar reasoning as for $s_{q_1}$, student $s_{q_2}$ is assigned in $G$ to a project $p_{t_4}$ such that $s_{q_2}$ prefers $p_{t_4}$ to $p_{t_3}$. Let $l_{z_2}$ be the lecturer who offers $p_{t_4}$. We are identifying a sequence $\langle s_{q_i}\rangle_{i \geq 0}$ of students, a sequence $\langle p_{t_i}\rangle_{i \geq 0}$ of projects, and a sequence $\langle l_{z_i}\rangle_{i \geq 0}$ of lecturers, such that, for each $i \geq 1$

\begin{enumerate}
\item $s_{q_{i}}$ prefers $p_{t_{2i}}$ to $p_{t_{2i-1}}$,
\item $(s_{q_i}, p_{t_{2i}}) \in G$ and $(s_{q_i}, p_{t_{2i - 1}}) \in M^*$,
\item $l_{z_i}$ prefers $s_{q_{i+1}}$ to $s_{q_{i}}$; also, $l_{z_i}$ offers both $p_{t_{2i}}$ and $p_{t_{2i+1}}$ (possibly $p_{t_{2i}} = p_{t_{2i+1}}$).
\end{enumerate}

Following a similar argument as in the proof of Lemma~\ref{pair-deletion-outer}, we can identify an infinite sequence of distinct students and projects, a contradiction.

Now suppose condition (b) is satisfied, i.e., some replete lecturer is not full in $M$. Let $L_1$ and $L_2$ be the set of replete and non-replete lecturers respectively. Then some $l_{k''} \in L_1$ satisfies $|M(l_{k''})| < d_{k''}$. Condition (a) cannot be satisfied, for otherwise the first part of the proof shows that $M^*$ does not exist. Hence $|M(l_{k})| = d_G(l_k) < d_k$ for all $l_k \in L_2$. Now $|M(l_k)| \leq d_k$ for all $l_k \in L_1$. Hence
\begin{align}
\label{ineq:replete-lecturer}
|M| = \sum_{l_k \in L_1} |M(l_k)| + \sum_{l_k \in L_2} |M(l_k)| < \sum_{l_k \in L_1} d_k + \sum_{l_k \in L_2} d_G(l_k) \enspace.
\end{align}
Now suppose that $|M^*(l_k)| = d_k$ for all $l_k \in L_1$, and $|M^*(l_k)| \geq d_G(l_k)$ for all $l_k \in L_2$. Then $|M^*| > |M|$ by \ref{ineq:replete-lecturer}, a contradiction. Hence either (i) $|M^*(l_{k'})| < d_{k'}$ for some $l_{k'} \in L_1$, or (ii) $|M^*(l_{k'})| < d_G(l_{k'})$ for some $l_{k'} \in L_2$. In case (ii), we reach a similar contradiction to that arrived at for condition (a). In case (i), we have that $l_{k'}$ is undersubscribed in $M^*$. As $l_{k'}$ is replete, there exists some student $s_i$ who was provisionally assigned to $l_{k'}$ during the algorithm's execution, but $s_i$ is not assigned to $l_{k'}$ in $M^*$. Moreover, there exists some project $p_{j} \in P_{k'}$ such that $(s_i, p_j) \notin M^*$ By Lemma \ref{stsp-deletion}, $s_i$ is not assigned to a project in $M^*$ that she prefers to $p_j$. Following a similar argument as above, we can identify a sequence of distinct students and projects, and as this sequence is infinite, we reach a contradiction. Thus $I$ admits no strongly stable matching.
\qed \end{proof}
\begin{restatable}[]{lemma}{studentmultiplybound}
\label{lemma:student-multiply-bound}
Suppose that in the final provisional assignment graph $G$, a student is bound to two or more projects that are offered by different lecturers.  Then $I$ admits no strongly stable matching.
\end{restatable}
\begin{proof}
Suppose that a strongly stable matching exists for the instance. Let $M$ be a feasible matching in the final provisional assignment graph $G$. Denote by $L_1$ and $L_2$ the set of replete and non-replete lecturers respectively. 

Denote by $S_1$ the set of students who are bound to one or more projects in $G$, and by $S_2$ the other students who are provisionally assigned to one or more projects in $G$. 
By Lemma \ref{student},
\begin{eqnarray}
\label{eqn:students}
|M| = |S_1| + |S_2|\enspace.
\end{eqnarray}
Also,
\begin{eqnarray}
\label{eqn:lecturers}
|M| = \sum_{l_k \in L_1} |M(l_k)| + \sum_{l_k \in L_2} |M(l_k)|\enspace.
\end{eqnarray}
Moreover, we have that \begin{eqnarray}
\label{eqn:repletenonrepletelecturers}
\sum_{l_k \in L_1} |M(l_k)| + \sum_{l_k \in L_2} |M(l_k)| = \sum_{l_k \in L_1} d_k + \sum_{l_k \in L_2} d_G(l_k);
\end{eqnarray}
for otherwise, no strongly stable matching exists by Lemma~\ref{replete-lecturer}. Now, if some student is bound to more than one project offered by different lecturers, by considering how the lecturers' quotas are reduced when the students in $S_1$ are removed in forming $G_r$ from $G$, it follows that 
\begin{eqnarray}
\label{eqn:bound-students-1}
\sum_{l_k \in L_1} (d_k - q_k^*) + \sum_{l_k \in L_2} (d_G(l_k) - q_k^*) &\geq& \sum_{l_k \in L_1} (d_k - q_k^*) + \sum_{l_k \in L_2} (\min\{d_G(l_k),\alpha_k\} - q_k^*) \nonumber \\
&=& \sum_{l_k \in L_1} (q_k - q_k^*) + \sum_{l_k \in L_2} (q_k - q_k^*)  \nonumber \\
&=& \sum_{l_k \in L_1 \cup L_2} (q_k - q_k^*) \nonumber \\
& > & |S_1|, \enspace
\end{eqnarray}
where $\sum_{l_k \in L_1} (d_k - q_k^*)$ and $\sum_{l_k \in L_2} (\min \{d_G(l_k), \alpha_k\} - q_k^*)$ is the number of students that are bound to at least one project offered by a lecturer in $L_1$ and $L_2$ respectively. By substituting Equality \ref{eqn:repletenonrepletelecturers} into Inequality \ref{eqn:bound-students-1}, we obtain the following
\begin{eqnarray}
\label{eqn:bound-students-2}
\sum_{l_k \in L_1} |M(l_k)| - \sum_{l_k \in L_1} q_k^* + \sum_{l_k \in L_2} |M(l_k)| - \sum _{l_k \in L_2}  q_k^*  > |S_1| \enspace .
\end{eqnarray}
Combining \ref{eqn:students} and \ref{eqn:lecturers} into \ref{eqn:bound-students-2}, we have 
\begin{eqnarray}
\label{eqn:bound-students-3}
|S_1| + |S_2| - \sum_{l_k \in L_1} q_k^* - \sum _{l_k \in L_2}  q_k^*  &>& |S_1|, \nonumber \\
\sum_{l_k \in L_1 \cup L_2} q_k^*  &<& |S_2|\enspace .
\end{eqnarray}
Since the total revised quota of lecturers in $L_1 \cup L_2$ whose projects are provisionally assigned to a student in $G_r$ is strictly less than the number of students in $S_2$, the preceding inequality suffices to establish that the critical set is non-empty, a contradiction.
 \qed \end{proof}
\begin{restatable}[]{lemma}{projectlemma}
\label{project}
Let $M$ be a feasible matching in the final provisional assignment graph $G$. Suppose that the pair $(s_i, p_j)$, where $p_j$ is offered by $l_k$, was deleted during the algorithm's execution. Suppose further that for any $p_{j'} \in P_k$ such that $s_i$ is indifferent between $p_j$ and $p_{j'}$, $(s_i, p_{j'}) \notin M$. Finally, suppose each of $p_j$ and $l_k$ is undersubscribed in $M$. Then $I$ admits no strongly stable matching.
\end{restatable}
 
\begin{proof}
Suppose for a contradiction that there exists a strongly stable matching $M^*$ in $I$. Let $(s_i, p_j)$ be a pair that was deleted during an arbitrary execution $E$ of the algorithm. This implies that $(s_i, p_j) \notin M^*$ by Lemma \ref{stsp-deletion}. Let $M$ be the feasible matching at the termination of $E$. By the hypothesis of the lemma, $l_k$ is undersubscribed in $M$. This implies that $l_k$ is undersubscribed in $M^*$, as proved in Lemma~\ref{lemma:lecturer-undersubscribed-tool}. Since $p_j$ is offered by $l_k$, and $p_j$ is undersubscribed in $M$, it follows from the proof of \cite[Lemma 5]{OM18} that $p_j$ is undersubscribed in $M^*$. Further, by the hypothesis of the lemma, either $s_i$ is unassigned in $M$, or $s_i$ prefers $p_j$ to $M(s_i)$ or is indifferent between them, where $M(s_i)$ is not offered by $l_k$. By Lemma \ref{stsp-deletion}, this is true for $s_i$ in $M^*$. Hence $(s_i, p_j)$ blocks $M^*$, a contradiction.
\qed \end{proof}

The next lemma shows that the feasible matching $M$ may be used to determine the existence, or otherwise, of a strongly stable matching in $I$.

\begin{restatable}[]{lemma}{nostsm}
\label{no-stsm}
Let $M$ be a feasible matching in the final provisional assignment graph $G$. If $M$ is not strongly stable then there is no strongly stable matching for the instance.
\end{restatable}
\begin{proof}
Suppose that $M$ is not strongly stable. Let $(s_i, p_j)$ be a blocking pair for $M$, and let $l_k$ be the lecturer who offers $p_j$. We consider two cases.
\begin{enumerate}
\item \emph{Suppose $s_i$ is unassigned in $M$ or prefers $p_j$ to $M(s_i)$}. Then $(s_i, p_j)$ has been deleted. 

Suppose $(s_i, p_j)$ was deleted as a result of $p_j$ being full or oversubscribed in line 8. Then $p_j$ is replete. Suppose firstly that $p_j$ ends up full in $M$, then $(s_i, p_j)$ cannot block $M$ irrespective of whether $l_k$ is undersubscribed or full in $M$, since $l_k$ prefers the worst assigned student/s in $M(p_j)$ to $s_i$. Thus $p_j$ is not full in $M$. As $p_j$ was previously full, each pair $(s_t, p_u)$, for each $s_t$ that is no better than $s_i$ at the tail of $\mathcal{L}_k$ and each $p_u \in P_k \cap A_t$, would have been deleted in line 29. Thus if $l_k$ is full in $M$, then $(s_i, p_j)$ does not block $M$. Now suppose $l_k$ is undersubscribed in $M$. If $l_k$ was full at some point during the algorithm's execution, then no strongly stable matching exists, by Lemma \ref{replete-lecturer}. Thus $l_k$ was never full during the algorithm's execution. This implies that $l_k$ is non-replete. To recap, we have that $p_j$ is a replete project that is not full in $M$, which is offered by a non-replete lecturer. Moreover, since $(s_i, p_j)$ has been deleted, no strongly stable matching exists, by Lemma \ref{project}. Thus $(s_i, p_j)$ was not deleted because $p_j$ became replete in line 8.

Now, suppose $(s_i, p_j)$ was deleted as a result of $l_k$ being full or oversubscribed in line 11.  Then $(s_i, p_j)$ could only block $M$ if $l_k$ ends up undersubscribed in $M$. If this is the case, then no strongly stable matching exists for the instance, by Lemma \ref{replete-lecturer}.

\vspace{0.1in}
Next, suppose $(s_i, p_j)$ was deleted because $p_j$ was a neighbour of some student $s_{i'} \in Z$ at a point when $s_{i}$ was in the tail of $\mathcal{L}_{k}^{j}$. Denote by $Z'$ the set of students in $Z$ who are provisionally assigned to $p_j$ in $G_r$. Suppose $p_j$ is non-replete, it follows that $0 < |Z'| \leq q_j^*$. Let $Z'' = Z \setminus Z'$. Then $\mathcal{N}(Z'') \subseteq \mathcal{N}(Z) \setminus \{p_j\}$, so that
$$\sum_{p_{j'} \in \mathcal{N}(Z'')} q_{j'}^* \leq \left(\sum_{p_{j'} \in \mathcal{N}(Z)} q_{j'}^*\right) - q_j^* \enspace.$$
Hence
\begin{eqnarray*}
\delta(Z'') &=& |Z''| - \sum_{p_{j'} \in \mathcal{N}(Z'')} q_{j'}^* \\
&=& |Z| - |Z'| - \sum_{p_{j'} \in \mathcal{N}(Z'')} q_{j'}^* \\
&\geq& |Z| - |Z'| - \left(\sum_{p_{j'} \in \mathcal{N}(Z)} q_{j'}^*\right) + q_j^* \\
&=& |Z| + q_j^* - |Z'| - \sum_{p_{j'} \in \mathcal{N}(Z)} q_{j'}^* \\
&\geq& |Z| - \sum_{p_{j'} \in \mathcal{N}(Z)} q_{j'}^* \\
&=& \delta(Z)
\end{eqnarray*}
If $\delta(Z'') > \delta(Z)$, then we reach a contradiction to the fact $Z$ is maximally deficient; hence $\delta(Z'') = \delta(Z)$. Moreover, $Z'' \subset Z$, contradicting the minimality of $Z$. Now, suppose $p_j$ is replete. The argument that no strongly stable matchings exists follows from the first paragraph. 

\vspace{0.1in}
Next, suppose $(s_i, p_j)$ was deleted (at line 29) because some other project $p_{j'}$ offered by $l_k$ is replete and subsequently becomes undersubscribed at line 23. Then $l_k$ must have identified her most preferred student, say $s_r$, rejected from $p_{j'}$ such that $s_i$ is at the tail of $\mathcal{L}_k$ and $s_i$ is no better than $s_r$ in $\mathcal{L}_k$. Moreover, every project offered by $l_k$ that $s_i$ finds acceptable would have been deleted from $s_i$'s preference list at the for loop iteration in line 29. Thus if $p_j$ ends up full in $M$, then $(s_i,p_j)$ does not block $M$. Suppose $p_j$ ends up undersubscribed in $M$. If $l_k$ is full in $M$, then $(s_i, p_j)$ does not block $M$, since $s_i \notin M(l_k)$ and $l_k$ prefers the worst student/s in $M(l_k)$ to $s_i$. Suppose $l_k$ is undersubscribed in $M$. Again, if $l_k$ is replete, then no strongly stable matching exists, by Lemma \ref{replete-lecturer}. Thus $l_k$ is non-replete. Since $(s_i, p_j)$ has been deleted, no strongly stable matching exists, by Lemma \ref{project}.

\item \emph{Suppose $s_i$ is indifferent between $p_j$ and $M(s_i)$}. First, suppose $p_j$ is undersubscribed in $M$. By the strong stability definition, $M(s_i)$ is not offered by $l_k$. Suppose $l_k$ is undersubscribed in $M$. If $l_k$ is replete, then by Lemma \ref{replete-lecturer}, no strongly stable matching exists. Hence $l_k$ is non-replete. First suppose $(s_i, p_j) \in G$, since $s_i \notin M(l_k)$, this implies that $l_k$ has fewer assignees in $M$ than provisional assignees in $G$. Thus no strongly stable matching exists by Lemma \ref{replete-lecturer}. Now suppose $(s_i, p_j) \notin G$. Then $(s_i, p_j)$ has been deleted. If $p_j$ is replete, then no strongly stable matching exists by Lemma \ref{project}. Hence $p_j$ is non-replete. To recap, $(s_i, p_j)$ has been deleted, and each of $p_j$ and $l_k$ is non-replete. There are two points in this algorithm where this deletion could have happened; (i) if $p_j$ was a neighbour of some student $s_{i'} \in Z$ at a point when $s_i$ was in the tail of $\mathcal{L}_k^j$, or (ii) at line 29 -- because some other project $p_{j'}$ offered by $l_k$ is replete and subsequently becomes undersubscribed at line 23. In any of the cases, the argument that no strongly stable matching exists follows from (1) above.
Finally, suppose $p_j$ is full in $M$.  Then $(s_i, p_j)$ cannot block $M$ irrespective of whether $l_k$ is full or undersubscribed in $M$, since the worst student/s in $M(p_j)$ are either better or no worse than $s_i$, according to $\mathcal{L}_k^j$.
\end{enumerate}

Since $(s_i, p_j)$ is an arbitrary pair, this implies that $M$ admits no blocking pair. Hence $I$ admits no strongly stable matching.
\qed \end{proof}
\begin{restatable}[]{lemma}{timestronglemma}
\label{time-complexity-strong}
\texttt{Algorithm SPA-ST-strong} may be implemented to run in $O(m^2)$ time, where $m$ is the total length of the students' preference lists. 
\end{restatable}

\begin{proof}
It is clear that the work done in the inner \texttt{repeat-until} loop other than in finding the maximum cardinality matchings and critical sets is bounded by a constant times the number of deleted pairs, and so is $O(m)$ (where $m$ is the total length of the students' preference lists). We remark that the total amount of work done outside the inner \texttt{repeat-until} loop (i.e., in deleting pairs when a replete project ends up undersubscribed in line 23) is bounded by the total length of the students' preference lists, and so is $O(m)$. 
During each iteration of the inner \texttt{repeat-until} loop of \texttt{Algorithm SPA-ST-strong}, we need to form the reduced assignment graph $G_r$, which takes $O(m)$ time. Further, we need to search for a maximum matching in $G_r$, which allows us to use Lemma~\ref{lemma:critical-set} to find the critical set. Next we show how we can bound the total amount of work done in finding the maximum matchings.

Suppose that \texttt{Algorithm SPA-ST-strong} finds a maximum matching $M_r^{(i)}$ in the reduced assignment graph $G_r^{(i)}$ at the $i$th iteration of the inner \texttt{repeat-until} loop. Suppose also that, during the $i$th iteration in this loop, $x_i$ pairs are deleted because they involve students in the critical set $Z^{(i)}$, or students tied with them in the tail of a project in $\mathcal{N}(Z^{(i)})$. Suppose further that in the $(i+1)$th iteration, $y_i$ pairs are deleted before the reduced assignment graph $G_r^{(i+1)}$ is formed. Note that any edge in $G_r^{(i)}$ which is not one of these $x_i + y_i$ deleted pairs must be in $G_r^{(i+1)}$, since an edge in the provisional assignment graph cannot change state from unbound to bound. In particular, at least $|M_r^{(i)}| - x_i - y_i$ edges of $M_r^{(i)}$ remain in $G_r^{(i+1)}$ at the $(i+1)$th iteration. Hence we can start from these edges in that iteration and find a maximum matching $M_r^{(i+1)}$ in $O(\min\{nm, (x_i+y_i+z_i)m\})$ time, where $z_i$ is the number of edges in $G_r^{(i+1)}$ which were not in $G_r^{(i)}$.

Suppose that a total of $s$ iterations of the inner \texttt{repeat-until} loop are carried out, and that in $t$ of these, $\min\{nm, (x_i+y_i+z_i)m\}) = nm$. Then the time complexity of maximum matching operations, taken over the entire execution of the algorithm, is
$O( \min\{n,\sum c_j\}m + tnm+ m\sum(x_i+y_i+z_i))\,$ where the first term is for the maximum matching computation in the first iteration, and the sum in the third term is over the appropriate $s - t  - 1$ values of $i$. We note that $\sum_{i=1}^s (x_i+y_i+z_i) \leq 2m$ (since each of the total number of deletions and provisional assignments is bounded by the total length of the students' preference lists), and since $x_i + y_i + z_i \geq n$ for the appropriate $t$ values of $i$, it follows that
$$tnm + m\sum(x_i+y_i+z_i) \leq m\sum_{i=1}^s (x_i+y_i+z_i) \leq 2m^2.$$
Thus
$$m\sum(x_i+y_i+z_i) \leq 2m^2 - tnm.$$

At the termination of the outer \texttt{repeat-until} loop, a feasible matching is constructed by taking the final maximum matching and combining it with the bound (student, project) pairs in the final provisional assignment graph. This operation is clearly bounded by the number of bound pairs, hence is $O(m)$. It follows that the overall complexity of \texttt{Algorithm SPA-ST-strong} is $O(m + \min\{n,\sum c_j\}m$  $ + 2m^2) = O(m^2).$

We remark that the problem of finding a maximum matching in $G_r$ is also called the \emph{Upper Degree Constrained Subgraph problem} ({\scriptsize UDCS}) \cite{Gab83}. An instance of $G_r$ can be viewed as an instance of {\scriptsize UDCS}, except that students have no explicit preferences in the {\scriptsize UDCS} case. Using Gabow's {\scriptsize UDCS} algorithm \cite{Gab83}, we can find a maximum matching in $G_r$ in time $O(|E_r|\sqrt{\sum q_j^*})$. However, it is not clear how this algorithm can be used to improve the overall running time of \texttt{Algorithm SPA-ST-strong}.
 \qed \end{proof}
 
\noindent
The following theorem collects together Lemmas~\ref{stsp-deletion}-\ref{time-complexity-strong} and establishes the correctness of \texttt{Algorithm SPA-ST-strong}.
\begin{theorem}
\label{thrm:spa-st-strong-optimality}
For a given instance $I$ of {\sc spa-st}, \texttt{Algorithm SPA-ST-strong} determines in $O(m^2)$ time whether or not a strongly stable matching exists in $I$. If such a matching does exist, all possible executions of the algorithm find one in which each assigned student is assigned at least as good a project as she could obtain in any strongly stable matching, and each unassigned student is unassigned in every strongly stable matchings.
\end{theorem}

Given the optimality property established by Theorem \ref{thrm:spa-st-strong-optimality}, we define the strongly stable matching found by \texttt{Algorithm SPA-ST-strong} to be \emph{student-optimal}. For example, in the {\sc spa-st} instance illustrated in Fig.~\ref{fig:spa-st-instance-3} (Page \pageref{fig:spa-st-instance-3}), the student-optimal strongly stable matching is $M = \{(s_1, p_6), (s_2, p_2), (s_4, p_5), (s_5, p_3), (s_6, p_4), (s_7, p_1), (s_8, p_1)\}$.

\section{Conclusion}
\label{sect:spa-st-strong-conclusions}
We leave open the formulation of a lecturer-oriented counterpart to \texttt{Algorithm SPA-ST-strong}. From an experimental perspective, an interesting direction would be to carry out an empirical analysis of \texttt{Algorithm SPA-ST-strong}, to investigate how various parameters (e.g., the density and position of ties in the preference lists, the length of the preference lists, or the popularity of some projects) affect the existence of a strongly stable matching, based on randomly generated and/or real instances of {\sc spa-st}. 


\section*{Acknowledgement}
The authors would like to convey their sincere gratitude to Adam Kunysz for valuable discussions and helpful suggestions concerning \texttt{Algorithm  SPA-ST-strong}.   They would also like to thank anonymous reviewers for their helpful suggestions.


\end{document}